%% file: main.tex
\documentclass[11  pt, oneside,onecolumn]{article} 
\usepackage[margin=1in,lmargin=1in,rmargin=1in,bmargin=1in,tmargin=1in]{geometry}  
\usepackage[latin1]{inputenc}
\usepackage{algorithm}
\usepackage{algorithmic}
\usepackage{amsthm}
\usepackage{natbib}
\usepackage{subfigure}
\RequirePackage{xcolor} 
\definecolor{mydarkblue}{rgb}{0,0.08,0.45}
\usepackage{hyperref}
\include{preamble}

\hypersetup{
	colorlinks=true, 
	linktoc=all,     
	linkcolor=blue,  
	citecolor = mydarkblue,
}
\renewcommand{\cite}[1]{\citep{#1}}

\author{Vicente Balmaseda$^{1}$, Ying Xu$^{2}$, Yixin Cao$^{2}$, Nate Veldt$^{1}$\\
	\\
	\normalsize{$^{1}$Department of Computer Science and Engineering, Texas A\&M University, Texas, USA}\\
	\normalsize{$^{2}$Department of Computing, The Hong Kong Polytechnic University, Hong Kong, China}\\
	\\
}

\begin{document}
	\title{Combinatorial Approximations for Cluster Deletion:\\Simpler, Faster, and Better}
	
	\date{}
	\maketitle
	\begin{abstract}
		Cluster deletion is an NP-hard graph clustering objective with applications in computational biology and social network analysis, where the goal is to delete a minimum number of edges to partition a graph into cliques. 
		We first provide a tighter analysis of two previous approximation algorithms, improving their approximation guarantees from 4 to 3.
		Moreover, we show that both algorithms can be derandomized in a surprisingly simple way, by greedily taking a vertex of maximum degree in an auxiliary graph and forming a cluster around it.
		One of these algorithms relies on solving a linear program.  Our final contribution is to design a new and purely combinatorial approach for doing so that is far more scalable in theory and practice.
	\end{abstract}

\input{content-arxiv/intro-v1}
\input{content-arxiv/prelims-v1}
\input{content-arxiv/approx-v2}
\input{content-arxiv/lp-solver-v2}
\input{content-arxiv/experiments-v1}

\section{Conclusion}
We have developed two combinatorial approximation algorithms for the \textsc{Cluster Deletion} problem, with a common degree-based pivoting strategy for derandomization.
We managed to show that both algorithms have an approximation guarantee of 3.
While the analysis for MatchFlipPivot is tight, it is not for the other.
One open question is whether the ratio can be improved by a tighter analysis.
Another major open problem is obtaining more efficient implementations of our algorithms.
For example, can we exploit the special structure of the minimum $s$-$t$ cut problem to avoid calling an off-the-shelf algorithm?
Another interesting direction for future work is to see whether we can adapt these ideas to obtain faster approximation algorithms for \textsc{Cluster Editing}.

\bibliography{main}
\bibliographystyle{plainnat}
	\appendix

\input{content-arxiv/appendix}

	
\end{document}

%% file: preamble.tex
\usepackage{amsfonts}
\usepackage{amsmath}
\usepackage{mathrsfs}  
\usepackage{graphicx}
\usepackage{enumitem}
\usepackage{adjustbox}
\usepackage{longtable}
\usepackage{amsfonts}
\usepackage{caption}
\usepackage{booktabs}
\usepackage{bbm}
\usepackage{blkarray}
\usepackage{circledsteps}

\usepackage{longtable}

\usepackage{xspace}

\usepackage{tikz,pgfkeys}
\usetikzlibrary{calc}

\newcommand{\cd}{\textsc{Cluster Deletion}\xspace}

\newcommand{\minstc}{\textsc{MinSTC}\xspace}

\usepackage{mathtools}

\newcommand{\wedges}{\ensuremath{\mathcal{W}}\xspace}

\usepackage{mathrsfs}  

\usepackage{amsthm}
\theoremstyle{plain}
\newtheorem{obser}{Observation}[section]
\newtheorem{theorem}{Theorem}[section]

\newtheorem{lemma}[theorem]{Lemma}

\theoremstyle{definition}

\theoremstyle{remark}

%% file: content-arxiv/intro-v1.tex
\section{Introduction}
Graph clustering is a fundamental task in graph mining where the goal is to partition nodes of a graph into disjoint clusters that have dense internal connections but are only sparsely connected to the rest of the graph. This has a wide variety of applications which include detecting communities in social networks~\cite{fortunato2010community}, identifying related genes in biological networks based on gene expression profiles~\cite{Ben-DorShamirYakhini1999}, and finding groups of pixels in an image that belong to the same object~\cite{shi2000normalized}. 
An idealized notion of a cluster in a graph is a set of nodes that is completely connected internally (i.e., a clique) while being completely disconnected from the rest of the graph. \emph{Cluster graph modification problems}~\cite{shamir2004cluster} are a class of graph clustering objectives that seek to edit the edges in a graph as little as possible in order to achieve this idealized structure. One widely studied problem is correlation clustering~\cite{BansalBlumChawla2004}, which 
can be cast as adding or deleting a minimum number of edges to convert a graph into a disjoint union of cliques. 
This problem is also known as \emph{cluster editing}. 
Designing approximation algorithms for different variants of correlation clustering has a long history, and has also seen extensive interest in the past few years in the machine learning community~\cite{jafarov2020ccasymmetric,jafarov2021local,bun2021differentially,cohen2021correlation,veldt2022correlation,stein2023partial,davies2023fast,assadi2023streaming}.

This paper focuses on a variant of correlation clustering called \cd{}, which seeks a minimum number of edges to \emph{delete} so that the graph becomes a disjoint set of cliques. \cd was first motivated by applications in clustering gene networks~\cite{Ben-DorShamirYakhini1999} and arises as an interesting special case of other more general frameworks for clustering~\cite{CharikarGuruswamiWirth2005,puleo2015correlation,veldt2018correlation}. The problem is NP-hard, but has been studied extensively from the perspective of parameterized algorithms~\cite{Gramm2005,damaschke2009bounded,gao2013cluster,bocker2011even,bathie2022kernels} and approximation algorithms~\cite{CharikarGuruswamiWirth2005,dessmark2007edgeclique,puleo2015correlation,veldt2018correlation,veldt2022correlation}. We provide several improved theoretical results and practical implementations for combinatorial algorithms for this task.

\paragraph{Previous work.}
The first approximation algorithm for \cd{} was based on rounding a linear programming (LP) relaxation and came with a factor $4$-approximation guarantee~\cite{CharikarGuruswamiWirth2005}. 
Other approximation algorithms based on the same canonical LP were subsequently developed~\cite{vanzuylen2009deterministic,puleo2015correlation}, culminating in the current-best approximation factor of 2~\cite{veldt2018correlation}. One limitation of all of these algorithms is that the underlying LP relaxation has $O(n^3)$ constraints for a graph with $n$ nodes, and is prohibitively expensive to solve in practice on large instances. Recently,~\citet{veldt2022correlation} provided faster approximation algorithms by rounding different and less expensive lower bounds for \cd{}. The first was a 4-approximation algorithm based on rounding an LP relaxation for a related problem called Strong Triadic Cluster (STC) labeling~\cite{sintos2014using}. The STC LP relaxation has fewer constraints than the canonical \cd{} LP relaxation, but still provides a lower bound on \cd{}. Veldt also developed the first \emph{combinatorial} approximation algorithm, called \emph{MatchFlipPivot}, which applies a fast algorithm for STC labeling and then rounds the resulting edge labels into a 4-approximate \cd{} solution. In numerical experiments, solving and rounding the STC LP relaxation using black-box LP software was shown to be roughly twice as fast solving and rounding the canonical LP, while \textit{MatchFlipPivot} was shown to be orders of magnitude faster. 

\paragraph{Motivating questions.} While these recent results lead to more practical algorithms, there is still a gap between theory and practice for \cd{} algorithms, and several open questions remain. Although the STC-based algorithms are faster and more practical, their 4-approximation guarantee is still noticeably worse than the 2-approximation based on the canonical LP relaxation. In practice, the STC-based algorithms tend to produce solutions that are much better than just a 4-approximation~\cite{veldt2022correlation}. 
A natural direction is to try to improve approximation factors and bridge the gap between theoretical and practical performance of STC-based methods. 

Another direction is to address the performance gap between \textit{MatchFlipPivot} and the STC LP rounding algorithm. \textit{MatchFlipPivot} is far faster in practice while satisfying the same worst-case approximation guarantee. At the same time, the STC LP relaxation is guaranteed to return a tighter lower bound for \cd{}, and was shown to produce higher quality results in practice. Furthermore, solving the STC LP relaxation was observed to often return the optimal solution for the canonical LP relaxation in practice. In these cases, the LP rounding technique is guaranteed to return a 2-approximate solution. These observations motivate the study of better approximation guarantees and faster techniques for solving the STC LP relaxation. 

Finally, existing implementations of the STC-based algorithms are randomized, and their approximation guarantees hold only in expectation. In theory these algorithms can be made deterministic by leveraging existing derandomization techniques~\cite{vanzuylen2009deterministic}. However, the deterministic versions are more complicated and slower, and as such have not been implemented in practice.

\paragraph{Our contributions.}
We significantly bridge the theory-practice gap by presenting algorithms that are \emph{simpler}, \emph{faster}, and have \emph{better} approximation guarantees.
\begin{itemize}[itemsep=0pt]
 \item  We provide a simplified presentation and a \emph{}{tight} analysis of the \textit{MatchFlipPivot} algorithm, proving an improved 3-approximation guarantee for the method and providing instances on which the ratio is asymptotically 3.
    \item We show a similar tighter analysis for an STC LP rounding algorithm, improving its approximation guarantee to 3.
    \item We improve the runtime of \textit{MatchFlipPivot} by designing a faster algorithm for a key step: computing a maximal edge-disjoint set of open wedges in a graph.
    \item We prove that the STC LP relaxation can be reduced to a minimum $s$-$t$ cut problem, leading to a faster, purely combinatorial version of our LP-based algorithm.
    \item We prove a simpler and faster new approach for \emph{deterministically} rounding a \cd{} lower bound into an approximate solution.
\end{itemize}

        
To put the last contribution into context, we note that previous approximations for \cd{} 
rely on (1)~computing a lower bound on a graph~$G$, (2)~rounding the lower bound into a new graph~$\hat{G}$, and (3)~forming clusters by \emph{pivoting} in~$\hat{G}$ (repeatedly select a node and cluster it with its neighbors). We prove that selecting pivot nodes based simply on degrees in~$\hat{G}$ provides the same approximation guarantee as other (more complicated and computationally expensive) deterministic pivoting strategies. 

We accompany our theoretical results with practical implementations and numerical experiments. They include the first implemented deterministic algorithms for \cd{}, which in practice produce solutions that are typically much less than 3 times the optimal solution. We also implement our combinatorial algorithm for solving the STC LP relaxation and demonstrate in practice that it is significantly faster than using black-box LP software and scales to instances that are orders of magnitude larger.

%% file: content-arxiv/prelims-v1.tex
\section{Preliminaries and Related Work}
Let \(G=(V,E)\) be an unweighted undirected graph with \(n=|V|\) and \(m=|E|\). We use the $\tilde{O}(\cdot)$ notation to suppress logarithmic factors in runtimes, e.g., $O(\log n) = \tilde{O}(1)$.
The problems we consider rely on the concept of open wedges. An \textit{open wedge centered at $k$} is a node triplet \((i,j,k)\) such that \((i,k) \in E\), \((j, k) \in E\) and \((i,j) \notin E\).
The third node indicates the center of the wedge. The order of the first two nodes in an open wedge is irrelevant, hence \((i,j,k) \equiv (j,i,k)\). 
Let $\wedges(G)$ be the set of open wedges in \(G\), and \(\wedges_k(G) \subseteq \wedges(G)\) be the set of open wedges centered at $k$. When $G$ is clear from context we simply write \wedges and $\wedges_k$.

\subsection{Cluster Deletion}
Given graph \(G\), \cd seeks a set of edges \(E_D \subseteq E\) that minimizes \(|E_D|\) such that \(G'=(V,E - E_D)\) is a disjoint set of cliques. This is equivalent to forming clusters in a way that minimizes the number of edges between clusters (known as ``mistakes'') while ensuring all clusters are cliques. This can be formulated as a binary linear program (BLP) as follows:
%
\begin{equation}
	\label{eq:cd_lp}
	\begin{array}{lll}
          \min & \displaystyle\sum_{(i,j) \in E} x_{ij}  &\\
          \text{s.t.} &	x_{ik} + x_{jk} \geq x_{ij} & \forall i,j,k\\
               & x_{ij} = 1	& \text{ if $(i,j) \notin E$} \\
               & x_{ij} \in \{0,1\} & \forall (i,j) \in E.
	\end{array}		
\end{equation}
This BLP has one variable for each pair of nodes, and \(x_{ij} = 0\) if and only if nodes \(i\) and \(j\) are in the same cluster. The {canonical} LP relaxation for \cd{} is obtained by replacing $x_{ij} \in \{0,1\}$ with nonnegativity constraints $x_{ij} \geq 0$. \citet{CharikarGuruswamiWirth2005} presented a 4-approximation based on this LP relaxation. The results of~\citet{vanzuylen2009deterministic} for constrained variants of correlation clustering imply a 3-approximation algorithm for \cd{} by rounding the same LP. The current best approximation factor for \cd, also obtained by rounding this LP, is 2~\cite{veldt2018correlation}.


\subsection{Strong Triadic Closure Labeling}
\cd{} has a well-documented connection to another NP-hard graph optimization problem~\cite{sintos2014using}. The latter problem is derived from the Strong Triadic Closure (STC) principle from social network analysis~\cite{granovetter1973strength,easley2010networks}, which states that if two individuals both have a strong connection to a mutual friend, they are likely to share at least a weak connection with each other.

Following this principle, we can label the edges in \(G\) as either weak or strong such that the STC principle is satisfied, i.e., each open wedge has at least one weak edge. This is called an \textit{STC labeling}, and is encoded by a set of weak edges $E_W \subseteq E$.
The \textit{minimum weakness strong triadic closure} (\minstc) problem~\cite{sintos2014using} is then the problem of finding a strong triadic closure labeling of \(G\) that minimizes the number of weak edges. Formally this is cast as the following BLP:
\begin{equation}\label{eq:minstc_blp}
	\begin{array}{lll}
		\text{min}          & \displaystyle\sum\limits_{(i,j) \in E} x_{ij} \\
		\text{s.t.}   & x_{ik}  + x_{jk} \geq 1 & \forall (i,j,k) \in \wedges\\
		& x_{ij} \in \{0,1\} & \forall (i,j) \in E. \\
	\end{array}
\end{equation}
The variable $x_{ij}$ is equal to 1 if and only if edge $(i,j)$ is a weak edge. The constraints in this BLP are in fact a subset of the constraints in the \cd{} BLP in~\eqref{eq:cd_lp}. This implies that every feasible solution $E_D$ for \cd defines a valid STC labeling $E_W = E_D$, and hence \minstc lower bounds \cd. However, deleting edges in an arbitrary STC labeling $E_W$ does not necessarily produce a disjoint union of cliques. The relationship between \minstc and \cd has been noted in several different contexts~\cite{konstantinidis2018strong,veldt2022correlation,bengali2023faster}, and there are known graphs where their optimal solutions differ by up to a factor of $8/7$~\cite{gruttemeier2020relation}. 

\textbf{Approximations based on vertex cover.} Solving \minstc over a graph $G = (V,E)$ is equivalent to finding a minimum vertex cover in the Gallai graph of $G$, obtained by associating each edge $(i,j) \in E$ with a vertex $v_{ij}$ and introducing an edge $(v_{ik}, v_{jk})$ in the Gallai graph if $(i,j,k)$ defines an open wedge in $G$~\cite{le1996gallai}. Every algorithm for vertex cover instantly implies an algorithm for \minstc with the same approximation factor. One simple 2-approximation for \minstc is to find a maximal edge-disjoint set of open wedges in $G$, then label an edge $(i,j) \in E$ as weak if it is in one of the open wedges in this set. This is equivalent to applying a standard maximal matching 2-approximation for vertex cover in the Gallai graph. Another simple 2-approximation is to solve the LP relaxation of the BLP in~\eqref{eq:minstc_blp} and label $(i,j) \in E$ as weak if $x_{ij} \geq 1/2$, analogous to a standard LP rounding algorithm for vertex cover. \citet{nemhauser1975vertex} showed that the LP relaxation for vertex cover is half integral, meaning that every basic feasible solution has LP variables satisfying $x_{ij} \in \{0, 1/2, 1\}$. This property therefore also holds for the STC LP relaxation, obtained by replacing binary constraints in~\eqref{eq:minstc_blp} with nonnegativity constraints $x_{ij} \geq 0$.

\subsection{STC + Pivot Framework}
\label{sec:stccd}
The {Pivot} algorithm repeatedly selects an unclustered node (the \emph{pivot}) in a graph and then clusters it with all of its unclustered neighbors. This was first designed as a way to approximate {correlation clustering}. When pivots are chosen uniformly at random and the procedure is applied directly to a graph $G$, this is a randomized 3-approximation algorithm for correlation clustering~\cite{AilonCharikarNewman2008}. Many algorithms for different variants of correlation clustering and \cd use Pivot as one step in a broader algorithmic pipeline~\cite{vanzuylen2009deterministic,ChawlaMakarychevSchrammEtAl2015,jafarov2020ccasymmetric,veldt2022correlation}. Choosing random pivots leads to approximation guarantees that hold only in expectation, but~\citet{vanzuylen2009deterministic} also showed techniques for carefully selecting pivot nodes in order to obtain deterministic approximation guarantees for different problems variants.

\citet{veldt2022correlation} recently provided a general framework for approximating \cd{} by combining STC labelings with pivoting procedures. The framework first (1)~obtains an approximately optimal STC labeling $E_W$ for a graph $G = (V,E)$, and then (2)~runs Pivot on graph $\hat{G} = (V, E - E_W)$ to form clusters. If pivoting on a node $k$ places two other nodes $i$ and $j$ inside a cluster, then both $(i, k)$ and $(j, k)$ are strong edges, which guarantees $(i,j) \in E$. This leads to a useful observation.
\begin{obser}
	\label{lem:cliques}
	If $E_W \subseteq E$ is an STC labeling for $G = (V,E)$, running Pivot on $\hat{G} = (V, E- E_W)$ with any pivot selection strategy produces clusters that are cliques in $G$.
\end{obser}
\citet{veldt2022correlation} used this framework to design two 4-approximation algorithms for \cd: one based on rounding the STC LP relaxation, and a faster purely combinatorial algorithm called \emph{MatchFlipPivot} based on finding a maximal edge-disjoint set of open wedges. The approximation guarantees hold in expectation when pivot nodes are chosen uniformly at random. The derandomized pivoting techniques of~\citet{vanzuylen2009deterministic} can be used to obtain deterministic approximation guarantees, though this is more involved conceptually and far slower computationally.


%% file: content-arxiv/approx-v2.tex
\section{Improved Approximation Analysis}
We prove tighter approximations and new deterministic rounding schemes for combining STC labelings with Pivot.

\begin{algorithm}[t]
	\caption{Pivot$(\hat{G} = (V,\hat{E}))$}
	\label{alg:piv}
	\begin{algorithmic}[1]
		\STATE $V' \gets V;  E'\gets \hat{E}; \mathcal{C} \gets \emptyset$
		\WHILE{$V'$ not empty}
		\STATE Select pivot $k \in V'$		 
		\STATE $C_k = k \cup \{i \in V' \colon (i,k) \in E'\}$ \hfill\texttt{// get neighbors}
		\STATE $\mathcal{C} = \mathcal{C} \cup \{C_k\}$  \hfill\texttt{// new cluster}
		\STATE $V' \leftarrow V' - C_k$;  $E' \leftarrow \hat{E} \cap (V' \times V')$	\hfill\texttt{// update graph}
		\ENDWHILE
		\STATE Return clustering $\mathcal{C}$
	\end{algorithmic}
\end{algorithm}
\subsection{Pivoting Lemma}
Algorithm~\ref{alg:piv} shows the generic Pivot algorithm applied to a graph $\hat{G}$. The resulting clusters are typically not cliques in $\hat{G}$, but we will combine these strategies with STC labeling techniques and Observation~\ref{lem:cliques} in order to design \cd approximation algorithms. Consider what happens if we have an induced subgraph $G' = (V',E')$ of $\hat{G}$ at some intermediate step of the Pivot algorithm and we pivot on a node $k \in V'$ to form a new cluster $C_k \subseteq V'$. Let $\mathit{deg}_k(G') = |C_k| - 1$ be the degree of node $k$ in $G'$ (the number of neighbors of $k$ in $V'$), and define two sets of node pairs:
\begin{align*}
	B_k(G') &= \{ (i,j) \in E' \colon (i,k) \in E', (j,k) \notin E'\}, \\
	N_k(G') &= \{ (i,j) \notin E' \colon (i,k) \in E', (j,k) \in E'\}.
\end{align*}
The set $B_k(G')$ represents edges on the \emph{boundary} of cluster $C_k$ and $N_k(G')$ is the set of \emph{non-edges} inside the cluster. We define three strategies for selecting pivots.  

\begin{itemize}
\item \textbf{Pivot Strategy 1.} Select a pivot $k$ with the maximum degree in $G'$.
\item \textbf{Pivot Strategy 2.} Select a pivot $k$ that minimizes $|B_k(G')| / |N_k(G ')|$.
\item \textbf{Pivot Strategy 3.} Select a pivot $k$ uniformly at random.
 \end{itemize}

\begin{lemma}
	\label{lem:piv}
	Let $\mathcal{B}$ be the set of edges between clusters and $\mathcal{N}$ be the set of non-edges inside clusters that result from running Algorithm~\ref{alg:piv}. If Pivot Strategy 1 or 2 is used, then $|\mathcal{B}| \leq 2|\mathcal{N}|$. If Pivot Strategy 3 is used, this holds in expectation: $\mathbb{E}[|\mathcal{B}| ]= 2 \mathbb{E}[|\mathcal{N}|]$.
\end{lemma}
\begin{proof}
	Consider the graph $G' = (V',E')$ at a fixed intermediate step of the algorithm. For an arbitrary node $v \in V'$  we write $N_v = N_v(G')$, $B_v = B_v(G')$ and $\mathit{deg}_v = \mathit{deg}_v(G')$ 
	
	\emph{Strategy 1 analysis.} Assume that node $k$ is chosen as the pivot when applying Pivot Strategy 1. For an arbitrary node $u \in C_k\setminus\{k\}$, 
let $b_u$ be the number of edges in $G'$ that are incident to $u$ but not contained in $C_k$, and let $n_u$ be the number of non-edges involving $u$ that are in $C_k$. 
	Note that
	\begin{align*}
		b_u + (|C_k| - 1) - n_u = \mathit{deg}_u \leq \mathit{deg}_k = |C_k| - 1,
	\end{align*}
	which implies that $b_u \leq n_u$. Each non-edge in $C_k$ involves two nodes from $C_k\setminus\{k\}$, so
	\begin{align*}
		|B_k| = \sum_{u \in C_k \setminus \{k\} } b_u \leq \sum_{u \in C_k \setminus \{k\} } n_u = 2|N_k|.
	\end{align*}
	Thus, the number of new boundary edges in each iteration is bounded by twice the number of non-edges in the new cluster. Summing across all iterations gives $|\mathcal{B}| \leq 2|\mathcal{N}|$.
	
\emph{Strategy 2 and 3 analysis.}
Let $\mathcal{W}'$ be the set of open wedges in $G'$.  The following four statements are equivalent: (1)~$(i, j, k)$ is an open wedge; (2)~$(i, k)\in B_{j}$, (3)~$(j, k)\in B_{i}$, and (4)~$(i, j)\in N_{k}$.  Thus, $\sum_{k \in V'} |N_k| = |\mathcal{W}'|$ and $\sum_{k \in V'} |B_k| = 2 |\mathcal{W}'|$.
In other words,
\begin{equation}
  \label{eq:step}
 \sum_{k \in V'} \left( |B_k| - 2|N_k|\right) =
    \sum_{k \in V'} |B_k| -
     2 \sum_{k \in V'} |N_k|= 
    0.
\end{equation}
Therefore, there is at least one node satisfying $|B_k| - 2|N_k| \leq 0$, and applying Pivot Strategy 2 guarantees that $|B_k| \leq 2 |N_k|$, so summing across iterations again gives $|\mathcal{B}| \leq 2|\mathcal{N}|$. Regarding Pivot Strategy 3, Eq.~\eqref{eq:step} implies that for a uniform random pivot, $\mathbb{E}_{k \in V'} [|B_k| - 2|N_k| ] = 0$. Thus, at every iteration, the expected number of new boundary edges is twice the expected number of non-edges inside the cluster. By linearity of expectation, $\mathbb{E}[|\mathcal{B}| ] = 2 \mathbb{E}[|\mathcal{N}|]$.
\end{proof}


\subsection{Rounding a Disjoint Open Wedge Set}\label{sec:mfp_3_aprox}
\begin{algorithm}[t]
	\caption{$\textsc{MatchFlipPivot}(G = (V,E))$}
	\label{alg:mfp}
	\begin{algorithmic}[1]
		\STATE $W \gets $ maximal edge-disjoint set of open wedges in $G$. 		
		\STATE $E_W \gets$ edges contained in some open wedge of $W$. 
		\STATE Form $\hat{G} = (V,E - E_W)$
		\STATE Run $\text{Pivot}(\hat{G})$ \hfill\texttt{// for some choice of pivot strategy}
	\end{algorithmic}
\end{algorithm}
One way to approximate \minstc over a graph $G = (V,E)$ is a straightforward adaptation of the matching-based approximation algorithm for vertex cover.  We find a maximal edge-disjoint set of open wedges $W \subseteq \mathcal{W}$, and then for each $(i,j,k) \in W$, place edges $(i,k)$ and $(j,k)$ into the weak edge set $E_W$. Note that $|W|$ is a lower bound for \minstc (and also \cd) since each open wedge in $W$ must contain at least one weak edge (or in the case of \cd, one deleted edge) and no two wedges in $W$ share an edge. The edge set $E_W$ is therefore a 2-approximation for \minstc since $|E_W| = 2|W|$.
%
The randomized \emph{MatchFlipPivot} (MFP) algorithm of~\citet{veldt2022correlation} 
runs Pivot on $\hat{G} = (V, E- E_W)$ with uniform random pivot nodes. The algorithm has an expected approximation ratio 4, and can be derandomized using the techniques of~\citet{vanzuylen2009deterministic} as a black box. We note here that this corresponds to running Algorithm~\ref{alg:mfp} using Pivot Strategy 2. 

Our next result improves on this prior work by providing a tighter analysis of MFP to show an improved approximation guarantee of 3. 
Furthermore, we prove that our simple new degree-based pivoting strategy also provides a deterministic 3-approximation. This is significant given that the bottleneck of the previous deterministic MFP algorithm was computing and updating $N_k$ and $B_k$ values.
\begin{theorem}
	\label{thm:mfp}
	When using Pivot Strategy 1 or 2 on $\hat{G}$, Algorithm~\ref{alg:mfp} is a deterministic $3$-approximation for \cd. When selecting pivots uniformly at random, it is a randomized 3-approximation algorithm.
\end{theorem}
\begin{proof}
Let $m_W$ denote the number of weak edges between clusters, and $m_S$ the number of other edges between clusters.  The three nodes in any open wedge $W\in \mathcal{W}$ must be separated into at least two clusters.  Thus, at least one of the two edges of every wedge in $W$ is between clusters.
This means that $m_{W} \ge |E_W|/2$.
Note that $m_S = |\mathcal{B}|$ and $|\mathcal{N}| = |E_W| - m_W$ because non-edges of $\hat G$ inside clusters are all weak edges.
By Lemma~\ref{lem:piv}, using Pivoting Strategy 1 or 2 on $\hat{G}$ guarantees that
\[
  m_S = |\mathcal{B}| \leq
  2 |\mathcal{N}| = 2(|E_W| - m_W).
\]
The total number of edges between clusters is thus:
\begin{align*}
  m_W + m_S &\leq m_W + 2(|E_W| - m_W) \\
            &= 2|E_W| - m_W \\
            &\leq \frac{3}{2} |E_W|\\
            &= 3 |W|\\
            &\leq 3 \text{OPT}_{CD}.
\end{align*}
If we select pivot nodes uniformly at random, $m_W$ and $m_S$ become random variables, but coupling Lemma~\ref{lem:piv} with the fact that $m_{W} \ge |E_W|/2$ for every choice of pivot nodes provides the same guarantee in expectation.
\end{proof}

The following theorem proves that independent of the pivot strategy used, Algorithm~\ref{alg:mfp} cannot have a ratio better than 3. 
\begin{theorem}
  \label{lem:mfp-best}
  The asymptotic ratio of Algorithm~\ref{alg:mfp} is at least three.
\end{theorem}
\begin{proof}
	For any even integer $n \ge 8$, we can construct a graph of $n$ vertices $\{v_{1}, v_{2}, \ldots, v_{n/2}, u_{1}, u_{2}, \ldots, u_{n/2}\}$.
	The vertex set $\{v_{1}, v_{2}, \ldots, v_{n/2}\}$ is a clique, and for each $i = 1, \ldots, n/2$, the vertex $u_{i}$ is adjacent to only $v_{i}$.  See Figure~\ref{fig:example} for the example when $n = 12$.
	The only optimal solution is
	\[
	E_D = \{(v_{1}, u_{1}), (v_{2}, u_{2}), \ldots, (v_{n/2}, u_{n/2})\}
	\]
	and it has cost $n/2$.
	On the other hand, it is easy to see that 
	\[
	W = \{(v_{1}, u_{2}, v_{2}), (v_{2}, u_{3}, v_{3}), \ldots, (v_{n/2}, u_{1}, v_{1})\}
	\]
	is a maximal edge-disjoint set of open wedges in $G$.
        The set $E_{W}$ is accordingly
        \[
          \{ (u_{1}, v_{1}), (u_{2}, v_{2}), \ldots, (u_{n/2}, v_{n/2}), 
          (v_{1}, v_{2}), (v_{2}, v_{3}), \ldots, (v_{n/2}, v_{1})\}.
        \]
	In $\hat G = (V, E - E_{W})$, for all $i=1,\ldots,n/2$, the vertex $u_{i}$ is isolated, and the vertex $v_{i}$ has $n/2 - 3$ neighbors.
	When applying Algorithm~\ref{alg:mfp}, whatever the pivot strategy is, the first pivot in $\{v_{1}, v_{2}, \ldots, v_{n/2}\}$ decides the solution.  The solution has cost
	\[
	n/2 + 2 ( n/2 - 2) = 3 n/2 - 4.
	\]
	Thus, the ratio is asymptotically three.
\end{proof}

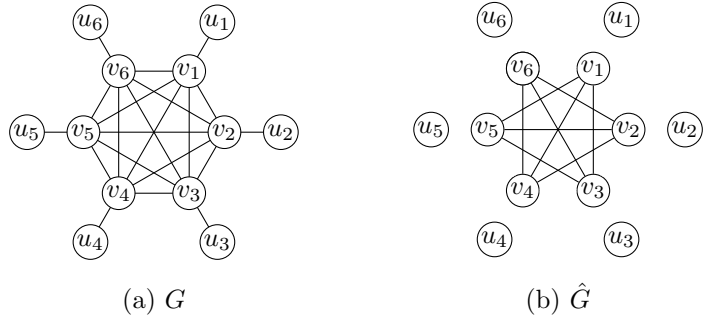
\begin{figure}[t]
	\tikzstyle{empty vertex}  = [{circle, draw, fill = white, inner sep=.5pt, minimum width=1.5pt}]
	\centering \small
	\begin{tikzpicture}[scale=.75]
		\def\n{6}
		\def\radius{1.25}
		\coordinate (v0) at ({90 + 180 / \n}:\radius) {};
		\foreach \i in {1,..., \n} {
			\pgfmathsetmacro{\angle}{90 - (\i - .5) * (360 / \n)}
			\coordinate (v\i) at (\angle:\radius) {};
			
			\draw(v\i) -- (\angle:{\radius+1.}) node[empty vertex] (u\i) {$u_{\i}$};
			\pgfmathsetmacro{\p}{\i - 1}        
			\foreach \j in {1, 2, ..., \p}
			\ifthenelse{\i>1}{\draw (v\i) -- (v\j)}{};
		}
		\foreach \i in {1,..., \n}
		\node[empty vertex] at (v\i) {$v_{\i}$};
		
		\node at (0, -3) {(a) $G$};      
	\end{tikzpicture}
	\qquad\qquad
	\begin{tikzpicture}[scale=.75]
		\def\n{6}
		\def\radius{1.25}
		\node[empty vertex] (v0) at ({90 + 180 / \n}:\radius) {$v_6$};
		\foreach \i in {1,..., \n} {
			\pgfmathsetmacro{\angle}{90 - (\i - .5) * (360 / \n)}
			\node[empty vertex] (v\i) at (\angle:\radius) {$v_\i$};
			
			\node[empty vertex] (u\i) at (\angle:{\radius+1.}) {$u_{\i}$};
			\pgfmathsetmacro{\p}{\i - 1}        
			\foreach \j in {1, 2, ..., \p}
			\ifthenelse{\i>1}{\draw (v\i) -- (v\j)}{};
		}
		\foreach \i in {1,..., \n} {
			\draw[white, ultra thick] let \n1 = {int(\i - 1)} in (v\n1) -- (v\i);
		}
		\node at (0, -3) {(b) $\hat G$};
	\end{tikzpicture}
	\caption{The example for Theorem~\ref{lem:mfp-best}.}
	\label{fig:example}
\end{figure}

\subsection{Rounding the STC LP Relaxation}
\label{sec:stc_lp_3_aprox}

\begin{algorithm}[t]
	\caption{\textsc{$\textsc{STC-LP-round}(G = (V,E))$}}
	\label{alg:stclp_rounding}
	\begin{algorithmic}[1]
		\STATE \(\{x_{ij}\}_{ij \in E} \gets\) (half-integral) solution to STC LP 
		\STATE Set \(E_W \gets \{ (i,j) \in E \colon x_{ij} \in \{1/2, 1\} \}\)
		\STATE 	Form $\hat{G} = (V,E - E_W)$
		\STATE Run $\text{Pivot}(\hat{G})$ \hfill\texttt{// for some choice of pivot strategy}
	\end{algorithmic}
\end{algorithm}
Recall that the LP relaxation of the BLP in~\eqref{eq:minstc_blp} provides a natural lower bound for both \minstc and \cd. Since every basic feasible solution of~\eqref{eq:minstc_blp} is half integral, we can obtain a set of variables $\{x_{ij}\}$ in polynomial time with $x_{ij} \in \{0, 1/2, 1\}$ for every $(i,j) \in E$, that minimizes~\eqref{eq:minstc_blp}. Given this solution, define $E_1 = \{(i,j) \in E \colon x_{ij} = 1\}$ and $E_h = \{(i,j) \in E \colon x_{ij} = 1/2\}$, and note that $E_W = E_1 \cup E_h$ defines an STC-labeling that is a 2-approximation for \minstc. We also refer to edges in $E_h$ as ``half-edges.'' Define $E_S = \{(i,j) \in E \colon x_{ij} = 0\} = E - E_W$ to be strong edges.
\citet{veldt2022correlation} showed that with randomized pivot nodes, Algorithm~\ref{alg:stclp_rounding} has an expected approximation ratio 4, and can be derandomized using the techniques of~\citet{vanzuylen2009deterministic}.\footnote{The deterministic STC-LP algorithm of~\citet{veldt2022correlation} incorporates LP values $\{x_{ij}\}$ more directly in choosing pivot nodes and is different from running Algorithm~\ref{alg:stclp_rounding} with Pivot Strategy 2. The latter provides a simplified and unified approach for rounding both types of \cd lower bounds we consider.} Mirroring Theorem~\ref{thm:mfp}, we provide an improved approximation analysis and show that our (simpler and faster) degree-based pivoting also gives a deterministic 3-approximation.
\begin{theorem}
	\label{thm:stclp}
	When using Pivot Strategy 1 or 2 on $\hat{G}$, Algorithm~\ref{alg:stclp_rounding} is a deterministic $3$-approximation for \cd. When selecting pivots uniformly at random, it is a randomized 3-approximation algorithm.
\end{theorem}
\begin{proof}
  We begin by proving that for every choice of pivot nodes in $\hat{G}$, at most half of the edges in $E_h$ will end up inside the clusters formed by Algorithm~\ref{alg:stclp_rounding}.
Consider an arbitrary choice of pivots.
Let ${\mathcal{B}}_h$ be the set of half-edges between clusters, and ${\mathcal{N}}_h$ the set of half-edges inside clusters.
We claim that the following set of variables is still feasible for the STC LP:
  \begin{align*}
    \hat{x}_{ij} = \begin{cases}
      x_{ij} & \text{ if $x_{ij} \in \{0,1\}$}, \\
      1 & \text{ if $x_{ij} \in {\mathcal{B}}$}_h, \\
      0  & \text{ if $x_{ij} \in {\mathcal{N}}$}_h.
    \end{cases}
  \end{align*}
  Consider an arbitrary open wedge $(i,j,k) \in \wedges$.
  Assume without loss of generality that $x_{ik} \ge x_{jk}$.
  If $x_{ik} = 1$, then $\hat x_{ik} = 1$ and $\hat{x}_{ik} + \hat{x}_{jk} \geq 1$.
  Otherwise, $x_{ik} = x_{jk} = 1/2$. 
  The three nodes in $W$ must be separated into at least two clusters.
  Thus, at least one of $x_{ik}$ and $x_{jk}$ is in $\mathcal{B}_h$, and hence $\hat{x}_{ik} + \hat{x}_{jk} \geq 1$.
  From the optimality of $\{{x}_{ij}\}$ and the feasibility of $\{\hat{x}_{ij}\}$ it follows that
  \begin{align*}
  \sum_{(i,j) \in E} x_{ij} \le \sum_{(i,j) \in E} \hat x_{ij} = \sum_{(i,j) \in E} x_{ij} + \frac{|\mathcal{B}_h|}{2} - \frac{|\mathcal{N}_h|}{2}. 
  \end{align*}
Thus, $|{\mathcal{B}}_h| \geq |{\mathcal{N}}_h|$ and $|{\mathcal{N}}_h| \leq (|{\mathcal{B}}_h| + |{\mathcal{N}}_h|) / 2 = |E_h|/2$. 
	
Let $m_1$ and $m_S$ be the numbers of edges between clusters that are from the sets $E_1$ and $E_{S}$, respectively.
Note that $|\mathcal{B}| = m_{S}$ and $|\mathcal{N}| = |E_1| - m_1 + |\mathcal{N}_h|$.
Using Pivot Strategy 1 or 2, Lemma~\ref{lem:piv} implies that $m_S \leq 2( |E_1| - m_1 + |\mathcal{N}_h|)$. We can then bound the total number of edges between clusters in $G$:
\begin{align*}
  m_1 + &|\mathcal{B}_h| + m_S \\
        &\leq m_1 + |\mathcal{B}_h| + 2(|E_1| - m_1 + |\mathcal{N}_h|) \\
        &= 2 |E_1| - m_1 + |\mathcal{B}_h| + 2 |\mathcal{N}_h| \\
        &\leq 2 |E_1| + \frac{3}{2} |E_h| \\
        &= \sum_{(i,j) \in E_1} 2x_{ij} + \sum_{(i,j) \in E_h} 3x_{ij} \\
        &\leq 3 \sum_{(i,j) \in E} x_{ij} \\
        & \leq 3 \text{OPT}_\mathit{CD}.
\end{align*}
Lemma~\ref{lem:piv} can similarly be used to show the same result in expectation for random pivot nodes.
\end{proof}


%% file: content-arxiv/lp-solver-v2.tex
\section{Faster Algorithms for Lower Bounds}
In addition to our improved approximation analysis, we provide faster algorithms for computing a maximal edge-disjoint sets of open wedges and for solving the STC LP. 

\subsection{Maximal Edge-Disjoint Open Wedge Set}
A simple existing approach for finding a maximal edge-disjoint open wedge set is to iterate through each node $k \in V$, and then iterate through pairs $\{i,j\}$ of neighbors of $v$. If $(i,j,k)$ is an open wedge and edge-disjoint from previously explored open wedges, we can add it to a growing set~$W$ of open wedges. We maintain a list~$E_W$ of edges that come from wedges in~$W$. This can be implemented in ${O}(\sum_{k \in V} d_k^2)$-time, which is always larger than $|\mathcal{W}|$ and can be as large as ${O}(nm)$.  While this is already fast in practice, we can further improve the theoretical runtime.
\begin{lemma}
\label{lem:owset}
A maximal edge-disjoint set of open wedges can be found in ${O}(m^{1.5})$ time and $O(m)$ space.
\end{lemma}
The appendix provides pseudocode and a full analysis for an algorithm satisfying these bounds. Similar to the previous approach, our procedure starts by iterating through nodes $k \in V$ and then iterates through pairs of neighbors of $k$. The key observation is that as soon as we encounter an open wedge $(i,j,k)$ and add its two edges $(i,k)$ and $(j,k)$ to $E_W$, we can effectively ``delete'' these edges and avoid exploring triplets involving them in future iterations. Since $|W|$ is always at most $m/2$, the total amount of work for adding edges to $E_W$ and then deleting them is bounded by $O(m)$. Now, when visiting two neighbors $\{i,j\}$ of $k$ we may find that $\{i,j,k\}$ is a triangle rather than an open wedge. We therefore have nothing to add to $W$ and no edges to delete. However, the amount of work that goes into finding these unwanted triangles is bounded in terms of the number of triangles in the graph, which is known to be at most $O(m^{1.5})$.
The result is inspired by the recent work of~\citet{cao2024breaking}, who provided the same type of bound for finding a maximal set of disjoint \emph{open triangles} in a complete signed graph. See the appendix for further details on the similarities and differences between these problems and approaches.

\subsection{Combinatorial Solver for the STC LP}\label{sec:stc_lp_combinatorial}
Although the STC LP rounding algorithm produces the same 3-approximation guarantee as MFP, the LP relaxation produces a tighter lower bound for \cd and tends to produce better solutions in practice. However, previous implementations rely on simply applying black-box LP solvers, which become a bottleneck both in terms of runtime and memory requirements~\cite{veldt2022correlation}. In this section we present a faster and purely combinatorial approach for solving the LP by reducing it to a minimum $s$-$t$ cut problem. This can be accomplished by first proving that the half-integral STC LP can be cast as a so-called monotone IP2 problem---an integer program with two variables per constraint with opposite signed coefficients. This can in turn be cast as a maximum closure problem~\cite{picard_maximal_1976} and then reduced to a maximum $s$-$t$ flow problem following the approach of~\citet{hochbaumApplications2021}. The reduction we present in this section merges several of these steps in order to provide a simplified and more direct reduction for the STC LP. 

\paragraph{Converting to BLP} 
The STC LP relaxation is obtained by replacing constraint $x_{ij} \in \{0,1\}$ in~\eqref{eq:minstc_blp} with $x_{ij} \ge 0$.
Since every basic feasible solution of the LP is half-integral, we can equivalently optimize over variables $x_{ij} \in \{0,1/2, 1\}.$ We will show that this is equivalent to the following BLP:
\begin{equation}\label{eq:ip}
	\begin{array}{lll}
		\min        & \displaystyle\sum\limits_{(i,j) \in E} \frac12 y_{ij} + \frac12(1- z_{ij}) & \\
		\text{s.t.}   & \left. \begin{aligned}
			z_{ik} \leq y_{jk} \\
			z_{jk} \leq y_{ik} \\
		\end{aligned} \right\} \qquad \forall (i,j,k) \in \wedges_k &\\
		& z_{ij}, y_{ij} \in \{0,1\} \quad \forall (i,j) \in E. &
	\end{array}
\end{equation}
%
We prove the following result.
\begin{lemma}
    \label{lem:equiv}
If $\{y_{ij}, z_{ij} \}_{ij \in E}$ is feasible for~\eqref{eq:ip}, then 
\begin{equation}\label{eq:ip2lp}
	x_{ij} = \frac12(y_{ij} - z_{ij} + 1) \quad \forall(i,j) \in E
\end{equation}
defines a feasible half-integral solution for the STC LP with the same objective value. Conversely, if $\{x_{ij}\}_{ij \in E}$ is a feasible half-integral solution to the STC LP, the variables
\begin{equation}
    \label{eq:LP2IP}
    (y_{ij}, z_{ij}) = 
    \begin{cases}
            (0,1) & \text{ if $x_{ij} = 0$} \\
            (1,1) & \text{ if $x_{ij} = 1/2$} \\   
            (1,0) & \text{ if $x_{ij} = 1$} 
    \end{cases}
    \quad \forall(i,j) \in E
    \end{equation} 
    are feasible for~\eqref{eq:ip} and have the same objective value.
\end{lemma}
\begin{proof}
	First, suppose that $\{y_{ij}, z_{ij} \}_{ij \in E}$ is feasible for \eqref{eq:ip}.  Then
	\begin{align*}
		x_{ik}  + x_{jk} &= \frac12(y_{ik} - z_{ik} + 1) + \frac12(y_{jk} - z_{jk} + 1)
		\\
          &= \frac12(y_{jk} - z_{ik}) + \frac12(y_{ik} - z_{jk}) + 1 \\
                                  &\ge 1.
	\end{align*}
	Thus, $\{x_{ij}\}_{ij \in E}$ is feasible for the STC LP. It is obviously half-integral and the objective value is the same.
	
	Now suppose that $\{x_{ij}\}_{ij \in E}$ is a feasible half-integral solution to the STC LP.  Consider an open wedge $(i, j, k)$.  Assume without loss of generality that $x_{ik} \ge x_{jk}$. 
	
	\begin{itemize}
		\item   If $x_{ik} = 1/2$, then $x_{jk} = 1/2$.  By \eqref{eq:LP2IP}, we have
		$y_{i k}= z_{jk} = 1$ and $y_{j k}= z_{ik} = 1$.
		\item   
		If $x_{ik} = 1$, then by \eqref{eq:LP2IP}, we have
		$y_{i k}= 1 \ge z_{jk}$ and $y_{j k} \ge z_{ik} = 0$.
	\end{itemize}
	Thus, $\{y_{ij}, z_{ij} \}_{ij \in E}$, which is integral, is feasible for \eqref{eq:ip}.  Since $\frac12(y_{ij} + 1 - z_{ij}) = x_{ij}$ for all the possible values of $x_{ij}$, the objective value is the same.
\end{proof}
We can therefore turn our attention to solving \eqref{eq:ip}, and use Eq.~\eqref{eq:ip2lp} to convert back to a solution to the STC LP.

\paragraph{Casting as a minimum $s$-$t$ cut problem}
The BLP in~\eqref{eq:ip} is equivalent to a minimum $s$-$t$ cut problem on a graph $G_{st}$. To construct the graph, we introduce a source node $s$ and a sink node $t$. Then for each $(i,j) \in E$, we introduce a node $Y_{ij}$ and a node $Z_{ij}$, then construct two directed edges $(s, Z_{ij})$ and $(Y_{ij}, t)$ each with weight $1/2$. Finally, for every $(i,j,k) \in \mathcal{W}_k$, we introduce two infinite-weight directed edges: $(Z_{ik}, Y_{jk})$ and $(Z_{jk}, Y_{ik})$.

Every $s$-$t$ cut in $G_{st}$ corresponds to a feasible solution to~\eqref{eq:ip}, where the weight of cut edges in $G_{st}$ equals the resulting objective score for~\eqref{eq:ip}. In more detail, nodes $Y_{ij}$ and $Z_{ij}$ correspond to binary variables $y_{ij}$ and $z_{ij}$ in~\eqref{eq:ip}. Setting a binary variable to 1 corresponds to placing its node on the $s$-side of the cut. For example, setting $y_{ij} = 1$ means placing $Y_{ij}$ on the $s$-side, which cuts the edge $(Y_{ij}, t)$. This contributes a $1/2$ to the cut penalty of $G_{st}$, just as setting $y_{ij} = 1$ contributes $1/2$ to the objective of~\eqref{eq:ip}. A similar penalty arises from setting $z_{ij} = 0$, which is equivalent to placing $Z_{ij}$ on the $t$-side of an $s$-$t$ cut in $G_{st}$. The infinite weight edges in $G_{st}$ encode the constraints of~\eqref{eq:ip}. The edge $(Z_{ik}, Y_{jk})$ has infinite weight to ensure that if $Z_{ik}$ is on the $s$-side of the cut, then $Y_{jk}$ is as well, just as $z_{ik} = 1$ forces $y_{jk} = 1$, required by the constraint $z_{ik} \leq y_{jk}$. 

To understand how this relates to \minstc, note that an edge $(s,Z_{ij})$ encourages $z_{ij}$ to be 1 and edge $(Y_{ij},t)$ encourages $y_{ij}$ to be zero. If both these preferences are satisfied, then $x_{ij} = 0$ meaning that $x_{ij}$ is a strong edge. If neither preference is satisfied, then $x_{ij} = 1$ and the edge is weak, whereas satisfying one preference but not the other leads to $x_{ij} = 1/2$. The reason these preferences typically cannot all be satisfied is because an edge $(Z_{jk}, Y_{ik})$ indicates that if $z_{jk} = 1$, this forces $y_{ik} = 1$.

\subsection{Runtime and Space Analysis} 
\label{sec:runtime}
We briefly summarize several improvements in runtime and space requirements that are obtained using our new techniques. More details for proving these bounds are included in the appendix. Lemma~\ref{lem:owset} improves the runtime for computing the MFP lower bound to ${O}(m^{1.5})$. This is always at least as fast as the previous ${O}(mn)$ runtime and is strictly faster for sparse graphs.
The previous deterministic pivoting scheme (Pivot Strategy 2) has space and runtime requirements that are $\Omega(m + |\mathcal{W}|)$, which can be $O(n^3)$ in the worst case, whereas degree-based pivoting can be implemented in $O(m)$ time and space. The most expensive step of Algorithm~\ref{alg:stclp_rounding} is solving the STC LP. Hence, the runtime for our combinatorial STC LP solver is also the asymptotic runtime for Algorithm~\ref{alg:stclp_rounding}. Using our reduction to minimum $s$-$t$ cut, we can can get a randomized solver that runs in $(m + |\mathcal{W}|)^{1 + o(1)}$ time by applying recent nearly linear time algorithms for maximum $s$-$t$ flows~\cite{chen2022maximum}. Using the algorithm of~\citet{goldberg1998beyond} we can get a deterministic algorithm with runtime $\tilde{O}(\min \{ (m + |\mathcal{W}|)^{1.5}, (m + |\mathcal{W}|)\cdot m^{2/3}$). For comparison, even the best recent theoretical solvers for general LPs would lead to runtimes that are $\Omega((m + |\mathcal{W}|)^2)$~\cite{van2020deterministic,shunhua2021faster,cohen2021solving}.


%% file: content-arxiv/experiments-v1.tex
\section{Experimental Results}
\begin{figure}[t]
	\centering
 \includegraphics[width=.9\linewidth]{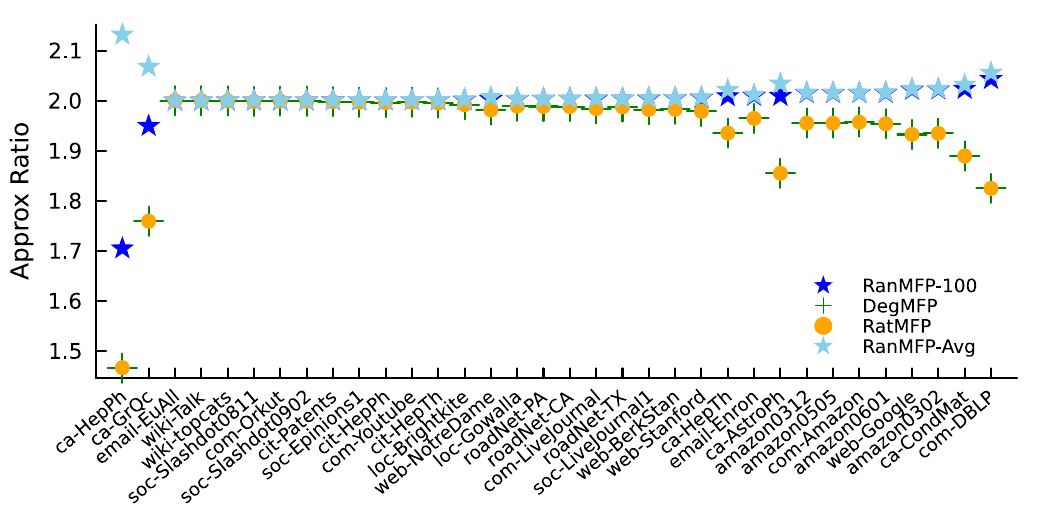}
 \vspace{-0\baselineskip} 
\caption{Approximation ratios ($|E_D|/|W|$) for MFP.}
	\label{fig:mfpratios}
\end{figure}
\citet{veldt2022correlation} previously showed experimental results for the \emph{randomized} variants of Algorithms~\ref{alg:mfp} and~\ref{alg:stclp_rounding} on a large collection of real-world graphs. For these experiments, the STC LP was solved using Gurobi optimization software. 
For our work we implement both deterministic schemes (Pivot Strategies 1 and 2) and compare them against each other and the randomized variant. We also show that our combinatorial approach for solving the STC LP is much faster and scales to much larger graphs than using a black-box LP solver.
Our algorithms are implemented in Julia.
Our experiments were run on a laptop with 16 GB of RAM.
For the most direct comparison with previous work~\cite{veldt2022correlation}, we consider the same collection of large graphs from the SNAP Repository~\cite{snapnets}, the largest of which (soc-Livejournal1) has 4.2 million nodes and 4.7 billion edges. We also run experiments on one even larger graph (com-Orkut) with
117.2 billion edges. The appendix includes additional details about datasets, implementations, and experimental results.

\textbf{Comparing pivot strategies.}
Our degree-based pivoting strategy is fast and practical. In addition to enjoying a deterministic approximation guarantee, it is comparable in speed to choosing random pivot nodes, while achieving much better approximations. Figure~\ref{fig:mfpratios} shows approximation ratios and Figure~\ref{fig:mfprun} shows runtimes achieved by MFP with degree-based pivots (DegMFP), pivots that minimize the ratio $|B_u|/|N_u|$ (RatMFP), and two different approaches to using random pivots. RanMFP-100 runs the random pivot strategy 100 times and takes the best solution found. This is a natural strategy to use since running randomized pivot once is very fast. RanMFP-Avg represents the average performance of the algorithm over these 100 trials.
\begin{figure}[t]
	\centering
	\includegraphics[width=.5\linewidth]{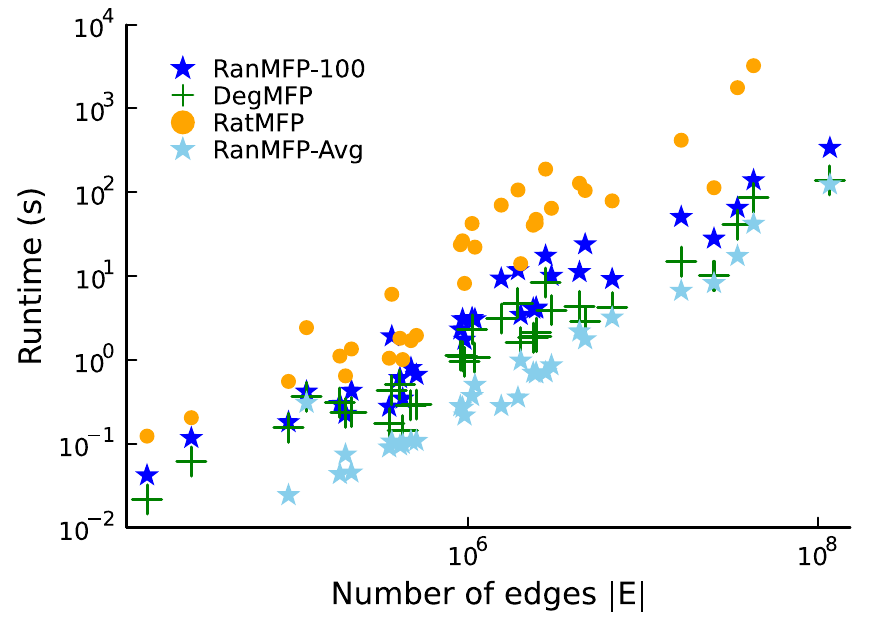}
	\caption{Runtimes of the MFP algorithms using different pivoting strategies. Each point represents one graph.}
	\label{fig:mfprun}
\end{figure}

DegMFP is almost identical to RatMFP in terms of approximation ratio, but is faster by an order of magnitude or more. DegMFP finds better solutions that RanMFP-100 and RanMFP-Avg, and is faster than RanMFP-100. Choosing random pivots once is faster (see runtimes for RanMFP-Avg), but especially for larger graphs DegMFP has comparable runtimes. A benefit of DegMFP is that by choosing high-degree nodes, it terminates in fewer pivot steps.

Most approximation ratios achieved by MFP are very close to 2. This can be explained by noting that the method labels a large percentage of edges as weak---between 63.4\% and 99.7\% for graphs in Figure~\ref{fig:mfpratios}. As a result, MFP deletes nearly all edges for some graphs, which is roughly a 2-approximation since $|E| \approx 2 |W|$. It is especially interesting to observe the behavior of different pivoting strategies when fewer edges are labeled weak and approximations factors deviate more from 2 (left- and right-most graphs in Figure~\ref{fig:mfpratios}). In these cases, RanMFP tends to have approximations that are worse than 2, while the deterministic schemes perform the best in these cases and detect more meaningful clusters. This highlights the utility of having a very fast deterministic rounding scheme for \cd algorithms. 

\begin{figure}[t]
	\centering
\includegraphics[width=.9\linewidth]{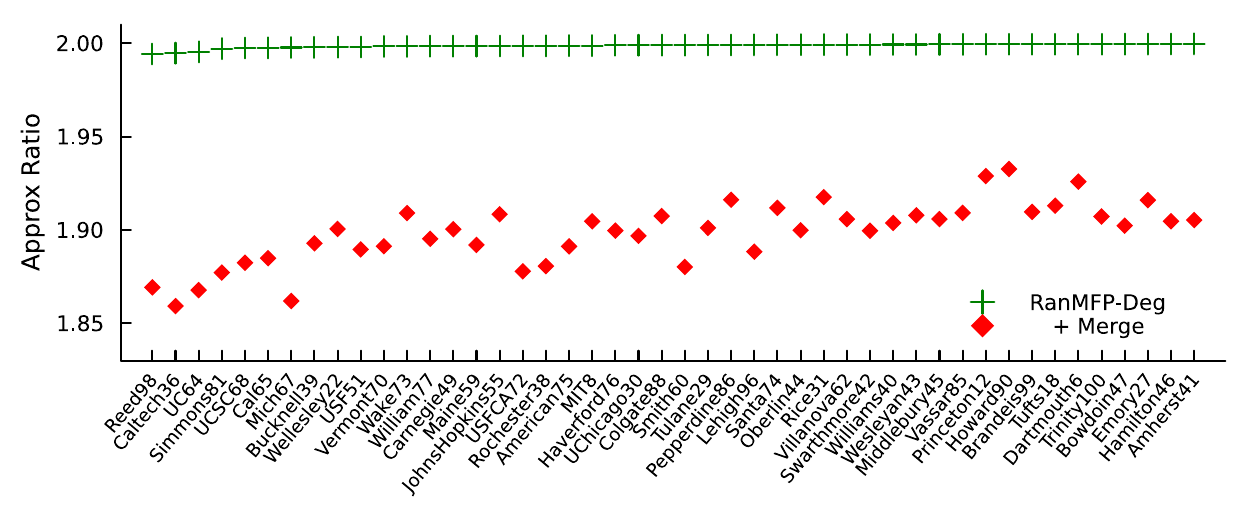}
\vspace{-0\baselineskip} 
	\caption{Improved approximation ratios when incorporating a cluster merging step after DegMFP.}
	\label{fig:mergeratios}
	\vspace{-0\baselineskip} 
\end{figure}
\textbf{Cluster merging heuristic.}
Figure~\ref{fig:mergeratios} shows results for DegMFP on graphs from the Facebook100 dataset~\cite{traud2012facebook} with up to 351k edges. For these graphs, finding an edge-disjoint open wedge set labels between 99.6\% and 99.95\% of edges as weak. MFP essentially achieves a 2-approximation by deleting nearly all edges. We also implement a heuristic for checking when clusters output by MFP can be merged into larger cliques (see appendix for details). 
Our current implementation is a proof of concept and not optimized for runtime. Nevertheless, this leads to noticeably better approximation ratios for these graphs, as well as others where a smaller percentage of edges are labeled weak (see appendix). Developing more scalable techniques for improving MFP is a promising direction for future research.


\begin{figure}[t]
	\centering
	\includegraphics[width=.5\linewidth]{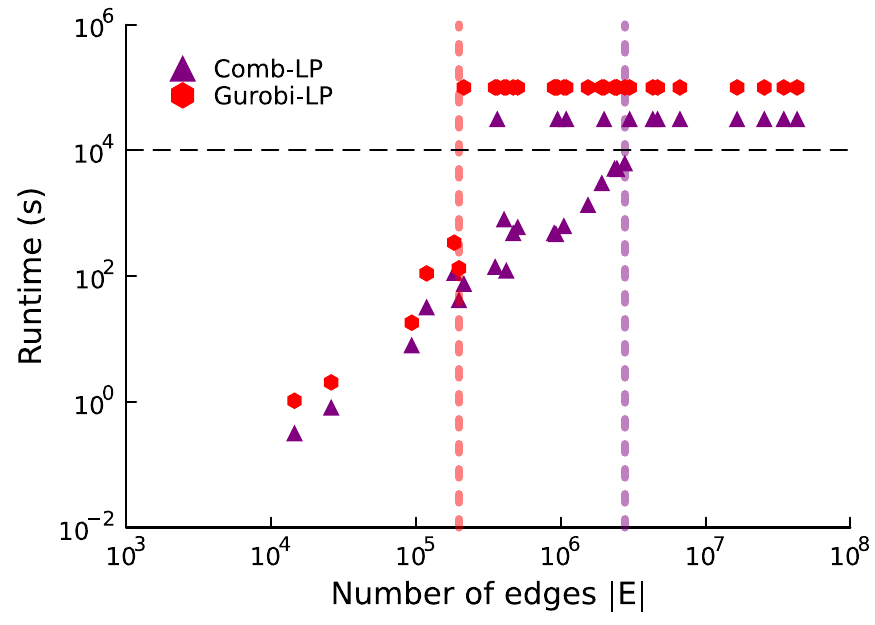}
	\caption{Runtimes of two different solvers for the STC LP. Each point represents a graph. Points above the black dashed line indicate graphs for which the given STC LP solver did not find a solution. The two vertical dashed lines indicate the size of the largest graph (in terms of edges) for which each method was able to successfully solve the LP.}
	\label{fig:lprun}
\end{figure}

\textbf{STC LP solvers.}
Our combinatorial solver for the STC LP enables us to run Algorithm~\ref{alg:stclp_rounding} more quickly and on a much larger scale than was previously possible. Figure~\ref{fig:lprun} shows runtimes for solving this LP using our combinatorial min-cut approach (Comb-LP) versus using general-purpose Gurobi optimization software (Gurobi-LP). The resulting objective score is the same for both since they are both finding an optimal solution for the LP. The main bottleneck of both algorithms is memory. For smaller graphs where both algorithms terminate without memory overflow, our combinatorial approach is roughly twice as fast. Overall, Gurobi-LP can only solve the STC LP on 6 of 34 graphs, while Comb-LP can solve it for 21 of 34.
The largest graph for which Gurobi finds a solution has 34,546 nodes and 420,877 edges, while our combinatorial approach was able to find solutions for many of the larger graphs, the largest of which has 1,971,281 nodes and 2,766,607 edges. Thus, in addition to being faster, this approach allows us to tackle problems that are an order of magnitude larger.

%% file: content-arxiv/appendix.tex
\section{Finding a Maximal Open Wedge Set}
\label{app:ow}
\begin{algorithm}[t]
	\caption{Faster maximal edge-disjoint open wedge set}
	\label{alg:fow}
	\begin{algorithmic}[1]
 \STATE 	{\bf Input.} Graph $G = (V,E)$
\STATE {\bf Output.} Maximal edge-disjoint set of open wedges $W \subseteq \mathcal{W}(G)$
\STATE $W = \emptyset$   \hfill\texttt{// edge-disjoint open wedge set}
\STATE $E_W = \emptyset$ \hfill \texttt{ //edges within W, will be labeled "weak"}
\FOR{$v \in V$} 
\STATE $S_v \leftarrow$ array of neighbors $u$ of $v$ such that $(u,v) \notin E_W$
\IF{$|S_v| < 2$}
\STATE continue \hfill \hfill\texttt{ // no open wedges centered at $v$}
\ENDIF
\STATE $\mathit{next} = [2, 3, \cdots , |S_v|, \text{NIL}]$  \hfill \texttt{ // $\mathit{next}[t] = $ neighbor of $v$ to visit after $S_v[t]$}
\STATE $i = 1$; $j = 2$; $j_\mathit{old} = 1$
\STATE $\mathit{finished} = \text{FALSE}$
\WHILE{$\mathit{finished} = \text{FALSE}$}
\STATE $u = S_v[i]$; $w = S_v[j]$ \hfill \texttt{// $(u, w, v)$ is the current wedge under consideration}
\IF{$(u, w) \in E$}
\STATE \texttt{// triangle found; do normal increment}
\STATE $[i,j,j_\mathit{old},\mathit{finished}] = \textsc{TriangleIncrement}(i,j,j_\mathit{old},\mathit{next}, \mathit{finished})$
\ELSE
\STATE \texttt{// open wedge found; do special increment}
\STATE $E_W \leftarrow E_W \cup \{(u,v), (w,v) \}$
\STATE $W \leftarrow W \cup \{(u,w,v) \}$
\STATE $[i,j,j_\mathit{old},\mathit{next},\mathit{finished}] = \textsc{OpenWedgeIncrement}(i,j,j_\mathit{old},\mathit{next}, \mathit{finished})$
\ENDIF
\ENDWHILE
\ENDFOR
\STATE \textbf{Return} $W$
\end{algorithmic}
\end{algorithm}

\begin{algorithm}[t]
	\caption{$[i,j,j_\mathit{old},\mathit{finished}] = \textsc{TriangleIncrement}(i,j,j_\mathit{old},\mathit{next}, \mathit{finished})$}
	\label{alg:triangleincrement}
	\begin{algorithmic}[1]
		\IF{$\mathit{next}[j] \neq \text{NIL}$}
		\STATE $j_\mathit{old} = j$ \hfill \texttt{// record $j_\mathit{old}$ such that $\mathit{next}[j_\mathit{old}] == j$}
		\STATE  $j = \mathit{next}[j]$ \hfill \texttt{ // increment $j$}
		\ELSIF{$\mathit{next}[j] == \text{NIL}$ and $\mathit{next}[i] \neq j$}
		\STATE $i = \mathit{next}[i]$\hfill \texttt{// increment $i$}
		\STATE $j = \mathit{next}[i]$  \hfill \texttt{// $i$ comes right after}
		\STATE $j_\mathit{old} = i$
		\ELSE
		\STATE $\mathit{finished} = \text{TRUE}$ \hfill \texttt{// done exploring neighbors since  $j = \mathit{next}[i]$ and $\mathit{next}[j] == \text{NIL}$}
		\ENDIF
		\STATE \textbf{Return} $i,j,j_\mathit{old},\mathit{finished}$
	\end{algorithmic}
\end{algorithm}

\begin{algorithm}[t]
	\caption{$[i,j,j_\mathit{old},\mathit{next},\mathit{finished}] = \textsc{OpenWedgeIncrement}(i,j,j_\mathit{old},\mathit{next}, \mathit{finished})$}
	\label{alg:owincrement}
	\begin{algorithmic}[1]
		\STATE $\mathit{next}[j_\mathit{old}] = \mathit{next}[j]$ \hfill \texttt{ // skip past $(v,v_j)$ in the future}
		\STATE $i = \mathit{next}[i]$ \hfill \texttt{// increment $i$}
		\IF{$i == \text{NIL}$ or $\mathit{next}[i] == \text{NIL}$}
		\STATE $\mathit{finished} = \text{TRUE}$
		\ELSE 
		\STATE $j = \mathit{next}[i]$
		\STATE $j_\mathit{old} = i$
		\ENDIF
		\STATE \textbf{Return} $i,j,j_\mathit{old},\mathit{next},\mathit{finished}$
	\end{algorithmic}
\end{algorithm}

\begin{algorithm}[t]
	\caption{Combinatorial solver for STC LP}
	\label{alg:comb-stc-lp}
	\begin{algorithmic}[1]
 \STATE 	{\bf Input.} Graph \(G = (V,E)\)
\STATE {\bf Output.} Optimal solution to STC LP
		\STATE Initialize graph $G_{st}$ with two nodes $V_{st} =\{s,t\}$
		\FOR{\((i,j) \in E\)}
            \STATE Add nodes $Z_{ij}, Y_{ij}$ to $V_{st}$
            \STATE Add edges $(s,Z_{ij})$ and $(Y_{ij},t)$ with weight $1/2$
            \ENDFOR
            \FOR{\((i,j,k) \in \mathcal{W}_k\)}
            \STATE Add edges $(Z_{ik}, Y_{jk})$ and $(Z_{jk}, Y_{ik})$ with weight $\infty$
            \ENDFOR
		\STATE Find min $s$-$t$ cut set $S \subseteq V_{st}$ of $G_{st}$
            \FOR{$(i,j) \in E$}
            \STATE Set $y_{ij} = \chi(Y_{ij} \in S)$, $z_{ij} = \chi(Z_{ij} \in S)$
		\STATE Set $x_{ij} = \frac{1}{2}(y_{ij} + 1 - z_{ij})$
            \ENDFOR
		\STATE \textbf{Return} \(\{x_{ij}\}_{(i,j) \in E}\)
	\end{algorithmic}
\end{algorithm}

Algorithm~\ref{alg:fow} is pseudocode for our ${O}(m^{1.5})$-time algorithm for finding a maximal edge-disjoint set of open wedges.

\textbf{Related work.} 
Our algorithm is loosely inspired by the recent work of~\cite{cao2024breaking}, who provided the same type bound for finding a maximal set of open triangles (i.e., $\{+, +, -\}$ triangles) in a complete signed graph. This is equivalent to finding a \emph{pair}-disjoint set of open wedges in a graph $G = (V,E)$, which is more restrictive than finding an \emph{edge}-disjoint set. The setting considered by~\citet{cao2024breaking} is also slightly different in that they associate lengths to different edges and restrict their search to open triangles satisfying certain length restrictions.
Nevertheless, two key ideas remain the same: (1)~once we add an open wedge to $W$, we never have to explore triplets of nodes involving its edges ever again, and (2)~the amount of time we spend exploring node triplets that turn out to be triangles can be bounded in terms of number of triangles in the graph, which is $O(m^{1.5})$.

\textbf{Algorithm explanation.} 
For our analysis we fix an ordering of the vertices $V = \{v_1, v_2, \hdots v_n\}$. We also assume that for each node $v_{i}, i=1,\ldots,n$, we have an array that stores the neighborhood of $v_{i}$ in this same order. The basic structure of Algorithm~\ref{alg:fow} is outlined in the main text: for each node $v \in V$, iterate through pairs of neighbors $\{u,w\}$ of $v$, checking each time if $(u,w,v)$ is an open wedge that can be added to a growing edge-disjoint open wedge set $W$. The key idea is that we must find a way to effectively ``delete'' the edges $(u,v)$ and $(w,v)$ if we add $(u,w,v)$ to $W$, so that we never waste time visiting another open wedge that involves either of these edges. This will ensure that we can charge all of the work done by the algorithm to an iteration where we add a new open wedge to $W$ (which happens $O(m)$ times), or to an iteration where we ``visit'' a triangle (which will happen $O(m^{1.5})$ times). 

When we first visit a node $v \in V$, we extract an array of neighbors $S_v$ (ordered by node label) with the property that for every node $u$ in $S_v$, $(u,v)$ is not in $E_W$ (the set of edges from the open wedge set $W$). This is the first way in which we ``delete'' edges adjacent to $v$ that we no longer wish to consider. As we iterate through pairs of nodes of $v$, we may encounter new open wedges, leading to other edges we must ``delete'' and skip over as we explore pairs of neighbors of $v$. We keep track of edges to skip over using an array $\mathit{next}$, where $\mathit{next}[t]$ represents the next index in $S_v$ to visit---in other words, the neighbor of $v$ that we should visit directly after visiting the neighbor $S_v[t]$. The array is initialized to $\mathit{next} = [2,3, \hdots, |S_v|, \text{NIL}]$ when we begin iterating through neighbors of $v$, which signifies that we default to visiting node $S_v[i+1]$ after visiting $S_v[i]$. Including NIL at the end of the array tells us when we have reached the end of $v$'s neighbor list.


We maintain two indices $i$ and $j$ that point to distinct neighbors of $v$: $S_v[i]$ and $S_v[j]$. If $\{S_v[i], S_v[j], v\}$ is a triangle, no edges are deleted and we just update indices $i$ and $j$ with help from $\mathit{next}$ to determine which new pair of neighbors of $v$ to explore next. If however this node triplet defines an open wedge, then we delete $(S_v[i], v)$ and $(S_v[j],v)$ by updating $\mathit{next}$ so that we skip over $S_v[i]$ and $S_v[j]$ (see Algorithm~\ref{alg:owincrement}) as we continue to iterate over pairs of neighbors of $v$. With this careful update, we ensure that every time we consider a new pair of neighbors of $v$, it defines either a new open wedge to add to $W$ or a triangle. 

\textbf{Runtime bound.} The time it takes to initialize $S_v$ and $\mathit{next}$ for each $v \in V$ is $O(d_v)$, where $d_v$ is the degree of $v$. The overall runtime for these steps (and all of the work done in lines 6--12 of Algorithm~\ref{alg:owincrement}) is therefore $O(m)$. The bottleneck of the algorithm comes from the while loop in lines 13--24. The total number of times we call \textsc{TriangleIncrement} (line 17) is $O(m^{1.5})$, since each triangle is visited at most 3 times (once for every node in the triangle) over the course of the algorithm. Each call to \textsc{TriangleIncrement} takes $O(1)$ time. We call \textsc{OpenWedgeIncrement} at most $O(m)$ times, which again takes $O(1)$ time each time. Finally, in order to figure out which of these two increment subroutines to call, we must check if there is an edge between two given neighbors of $v$ (line 15). This can be done in $O(1)$ time and $O(n^2)$ space by storing the full adjacency matrix. We improve this to $O(1)$ time and $O(m)$ space by storing all of the edges in a hash table that is queried to check adjacency.

\section{Theoretical Runtime Analysis}
\label{app:runtime}

Here we provide additional details behind the runtime and space analysis in Section~\ref{sec:runtime}. 

\paragraph{MFP algorithms.}
Appendix~\ref{app:ow} provided a detailed analysis to show that we can find a maximal edge-disjoint set of open wedges in ${O}(m^{1.5})$ time and $O(m)$ space. This is always at least as fast as the previous runtime of ${O}(mn)$ and is strictly better for sparse graphs. Computing this lower bound is the most expensive step for the randomized variant of MFP, so our work improves the best runtime for this randomized variant, while maintaining the same $O(m)$ space requirement.

For deterministic MFP, our runtime (and space) improvement is even more dramatic, thanks to our simplified deterministic rounding scheme. The previous deterministic rounding scheme (Pivot Strategy 2) requires identifying all open wedges and storing a map back and forth between open wedges and edges they contain. This means that the space and runtime requirements are both $\Omega(|W| + m)$. In contrast, pivoting based on degrees can be implemented in $O(m)$ time and space. This is done by computing node degrees and then placing nodes in one of $n$ bins based on their degree. We can then iterate through bins, greedily selecting a node with the highest degree and forming a cluster around it. For each node $u$ added to a cluster, we can update the degrees of its neighbors, and update their location in the bins, in $O(d_u)$ time. Overall the runtime is $O(m)$. With this new approach, we now have randomized and deterministic variants of MFP with a runtime of ${O}(m^{1.5})$ and space constraint of $O(m)$.

\paragraph{STC LP solvers.}
Algorithm~\ref{alg:comb-stc-lp} is the pseudocode for the combinatorial procedure for solving the STC LP relaxation given in Section~\ref{sec:stc_lp_combinatorial}.
When finding a minimum $s$-$t$ cut of $G_{st}$, we can double the weight of edges adjacent to $s$ and $t$ and replace infinite weight edges with edges of weight $m+1$, without changing the optimal solution. This means we are solving an $s$-$t$ cut problem in a directed graph with $2m + 2$ nodes and and $2m + 2|\mathcal{W}|$ edges, where all edge weights are integers that are polynomially bounded in the input size. Using the nearly linear time maximum $s$-$t$ flow solver of~\citet{chen2022maximum}, the STC LP relaxation can be solved in $(m + |\mathcal{W}|)^{1 + o(1)}$-time. We can also obtain a runtime of $\tilde{O}(|\mathcal{W}| + m^{1.5})$ using the results of~\citet{van2021minimum}, which is slightly better if $m^{1.5} = O(|\mathcal{W}|)$. We get a deterministic algorithm with a runtime of $\tilde{O}(\min \{ (m + |\mathcal{W}|)^{1.5}, (m + |\mathcal{W}|)\cdot m^{2/3}$) by using the well-known maximum $s$-$t$ flow algorithm of~\citet{goldberg1998beyond}. Note that the relative size of $(m + |\mathcal{W}|)^{1/2}$ and $m^{2/3}$ depends on the graph. For example, it is possible for $|\mathcal{W}|$ to be very small (even zero), in which case the former is better. It is also possible to have $|\mathcal{W}| = \Theta(n^3)$ and $m = \Theta(n^2)$, in which case $m^{2/3}$ is smaller.

There are several recent theoretical algorithms for solving an LP of the form $\min_{\textbf{A}\textbf{x} = \textbf{b}; \textbf{x} \geq 0} \textbf{c}^T \textbf{x}$, that run in current matrix multiplication time~\cite{van2020deterministic,shunhua2021faster,cohen2021solving}. When written in this format, the STC LP has $(m + |\mathcal{W}|)$ constraints, and applying these solvers yields runtimes that are all at least $\Omega((m + |\mathcal{W}|)^2)$.

\section{Additional Implementation and Experimental Details}

\paragraph{Implementation details}
Section~\ref{sec:runtime} and Appendix~\ref{app:runtime} provide details for the best runtimes that can be obtained in theory. As is often the case, there is a gap between the best theoretical algorithms and the most practical approaches. As a key example, the current best algorithms for maximum $s$-$t$ flows do not come with practical implementations. For solving minimum $s$-$t$ cut problems in practice, we instead apply a fast Julia implementation of the push-relabel algorithm for maximum $s$-$t$ flows. When computing a maximal edge-disjoint set of open wedges in practice, we apply the simpler approach that iterates through nodes and then pairs of neighbors of nodes to check for open wedges. In our experimental results, computing this lower bound tends to be very fast, and is often even faster than the pivot step.

The example in the proof of Theorem~\ref{lem:mfp-best} motivates the following postprocessing step.
After a set of clusters is obtained, we check each pair of them.  If all the edges between two clusters are present in the input graph, we merge them into a single cluster.
If the algorithm above has produced $c$ clusters, this step may take $\Omega(c^{2})$ time, which can be prohibitive.
Since this preprocessing can be stopped anytime safely, in practice, we can apply it with a fixed time. 
It can be completely turned off to save time, or carried out exhaustively to achieve the best effect.
Our experiments, shown in Table~\ref{tab:small} and Figure~\ref{fig:mergeratios}, demonstrate that for most graphs, it significantly improves the outcomes.
It is worth noting that this postprocessing is independent on the algorithms used to produce the clusters, and can be used for other algorithms for \textsc{Cluster Deletion}.

All of our algorithms are implemented in Julia. 
Code for our algorithms and experiments can be found at~\url{https://github.com/vibalcam/combinatorial-cluster-deletion}. Given the overlap with the previous work, our new implementations build directly on and improve the open source implementations of~\citet{veldt2022correlation}, available at~\url{https://github.com/nveldt/FastCC-via-STC}, released under an MIT License. 

\paragraph{Datasets}
The graphs in Figure~\ref{fig:mfpratios} all come from the SNAP network repository~\cite{snapnets} and come from several different types of graph classes. This includes social networks, road networks, citation networks, collaboration networks, and web graphs. All graphs have been standardized to remove weights, directions, and self-loops. The number of nodes and edges for each dataset are shown in Table~\ref{tab:full_results}.
The graphs in Figure~\ref{fig:mergeratios} are from the Facebook100 dataset~\cite{traud2012facebook}. We specifically consider 46 smallest graphs in terms of the number of edges. The largest of these is Cal65, which has 11,247 nodes and 351,358 edges. 
Running MFP is very fast on all of these graphs (always less than 0.02 seconds for DegMFP). We only considered the smallest graphs in the collection since (in its current form) our cluster merging heuristic does not scale as easily (taking over 2 hours for the largest graph). 
The small graphs considered in Table~\ref{tab:small} are all available from the Suitesparse matrix collection~\cite{suitesparse2011davis}. We specifically consider graphs from the Arenas collection (email, celegansmetabolic), Newman collection (Netscience, polblogs, polbooks, football, adjnoun, celegansneural), Pajek collection (Erdos991, Roget, SmaGri), and Mathworks collection (Harvard500).

\paragraph{Additional experimental results}
\begin{table}[t!]
	\caption{Improved approximation ratios (\cd{} solution divided by MFP lower bound) for several small graphs when using the merge heuristic. The \% weak column reports the percentage of edges that MFP labels weak before performing the pivot step.}
	\label{tab:small}
	\centering
	\scalebox{0.95}{\begin{tabular}{l r r lll}
			\toprule
			\textbf{Graph} & $|V| $& $|E|$ & \% weak & MFP approx & \emph{+merge} approx\\
   \midrule
Netscience & 379 & 914 & 68.93 & 1.72& 1.58 \\
Harvard500 & 500 & 2043 & 75.97 & 1.95& 1.79 \\
football & 115 & 613 & 83.52 & 1.79& 1.47 \\
Erdos991 & 446 & 1413 & 95.4 & 1.99& 1.67 \\
celegansmetabolic & 453 & 2025 & 95.41 & 1.99& 1.78 \\
polbooks & 105 & 441 & 95.69 & 1.98& 1.68 \\
email & 1133 & 5451 & 95.98 & 1.98& 1.80 \\ 
SmaGri & 1024 & 4916 & 98.05 & 2.00& 1.85 \\
Roget & 994 & 3640 & 98.24 & 2.00& 1.69 \\
celegansneural & 297 & 2148 & 98.88 & 2.00& 1.81 \\
adjnoun & 112 & 425 & 99.29 & 2.00& 1.76 \\
polblogs & 1222 & 16714 & 99.75 & 2.00& 1.92 \\
			\bottomrule
	\end{tabular}}
\end{table} 

 In Table~\ref{tab:small} we display improved approximation ratios that can be achieved using our cluster merging heuristic on several small graphs. Although our current implementation of the algorithm is not optimized for runtime, it leads to noticeably better approximation ratios for all graphs where we ran it. 

In Table~\ref{tab:full_results} we provide a more detailed look at the performance of our algorithms on the large SNAP graphs. This includes runtimes for computing lower bounds, the value of lower bounds, and the a posteriori approximation ratios (\cd{} solution divided by the lower bound computed by each method) achieved using our STC LP rounding method and MFP with three different pivoting strategies. For rounding the STC LP, we use the degree-based pivoting strategy as this is both fast and deterministic.

It is interesting to note that rounding the output to our combinatorial STC LP solver (the Comb-LP column in Table~\ref{tab:full_results}) tends to produce worse approximation ratios than MFP because it produces worse \cd{} solutions. This may be because the STC LP rounding technique is somewhat simplistic in that it treats all edges in $E_h = \{(i,j) \in E \colon x_{ij} = 1/2\}$ and $E_1 = \{(i,j) \in E \colon x_{ij} = 1\}$  in the same way (namely, delete these edges and then perform a pivoting step). One open question is whether it is possible to develop an improved rounding scheme that better leverages the difference between edges in $E_h$ and edges in $E_1$. As a related observation, rounding the STC LP solution returned by Gurobi tends to produce better results, beating MFP in some cases. This is likely due to slight differences in \emph{which} solution each method finds for the STC LP. An interesting direction for future research is to better understand how different solutions to the LP affect downstream \cd solutions, and then design techniques for quickly finding the most favorable LP solution for \cd. 

It is worth noting that despite the lower approximation ratios produced by Comb-LP, the STC LP lower bound for \textsc{Cluster Deletion} is always at least as tight as the lower bound computed by MFP. Furthermore, this lower bound is strictly tighter for all instances observed in practice. Thus, by combining the output of MFP with the tighter STC LP lower bound, we immediately obtain strictly better approximation guarantees than we obtain by running either method by itself.

\input{content-arxiv/table-results}

%% file: content-arxiv/table-results.tex


\begin{longtable}{llrrrrr}
\caption{Detailed results (\(n\) = \(|V|\), and \(m\) = \(|E|\)). Asterisks indicate the method ran out of memory. Dashes indicate one case where Gurobi failed to produce an optimal solution for another reason. method ran out of time (2 days maximum allotted time). The RanMFP column shows results specifically for RanMFP-100.}
\label{tab:full_results}\\
\toprule
\textbf{Graph} & & \textbf{RanMFP}& \textbf{DegMFP}& \textbf{RatMFP}& \textbf{Comb-LP}& \textbf{Gur-LP}\\*
\midrule
\endfirsthead
\toprule
\textbf{Graph} & & \textbf{RanMFP}& \textbf{DegMFP}& \textbf{RatMFP}& \textbf{Comb-LP}& \textbf{Gur-LP}\\*
\midrule
\endhead
\multicolumn{7}{|c|}{\textit{Continued on next page}} \\
\hline
\endfoot
\endlastfoot
\textsc{amazon0302} & LB & 402,933 & 402,933 & 402,933 & 438,400.0 & ***\\*
& UB & 814,852 & 779,652 & 779,523 & 855,620.0 & ***\\*
$n =$ 262,111 & Ratio  & 2.022 & 1.935 & 1.935 & 1.952 & ***\\* [2mm]
$m =$ 899,792 & Run LB  & 0.262 & 0.196 & 0.282 & 488.675 & ***\\*
& Run round  & 2.043 & 0.927 & 23.413 & 0.848 & ***\\*
& Run total  & 2.305 & 1.123 & 23.694 & 489.523 & ***\\*
\midrule
\textsc{amazon0312} & LB & 1,099,863 & 1,099,863 & 1,099,863 & 1.1616155e6 & ***\\*
& UB & 2,216,290 & 2,150,602 & 2,150,294 & 2.292491e6 & ***\\*
$n =$ 400,727 & Ratio  & 2.015 & 1.955 & 1.955 & 1.974 & ***\\* [2mm]
$m =$ 2,349,869 & Run LB  & 0.673 & 0.608 & 0.685 & 5,170.367 & ***\\*
& Run round  & 3.406 & 1.231 & 39.710 & 0.166 & ***\\*
& Run total  & 4.079 & 1.839 & 40.395 & 5,170.533 & ***\\*
\midrule
\textsc{amazon0505} & LB & 1,141,860 & 1,141,860 & 1,141,860 & 1.205665e6 & ***\\*
& UB & 2,300,995 & 2,232,537 & 2,232,231 & 2.378196e6 & ***\\*
$n =$ 410,236 & Ratio  & 2.015 & 1.955 & 1.955 & 1.973 & ***\\* [2mm]
$m =$ 2,439,437 & Run LB  & 0.660 & 0.812 & 0.645 & 5,204.671 & ***\\*
& Run round  & 3.435 & 1.290 & 46.796 & 0.152 & ***\\*
& Run total  & 4.095 & 2.102 & 47.441 & 5,204.823 & ***\\*
\midrule
\textsc{amazon0601} & LB & 1,142,262 & 1,142,262 & 1,142,262 & 1.2069215e6 & ***\\*
& UB & 2,302,232 & 2,232,009 & 2,231,696 & 2.37985e6 & ***\\*
$n =$ 403,394 & Ratio  & 2.016 & 1.954 & 1.954 & 1.972 & ***\\* [2mm]
$m =$ 2,443,408 & Run LB  & 0.669 & 0.618 & 0.672 & 5,102.761 & ***\\*
& Run round  & 3.535 & 1.258 & 41.885 & 0.180 & ***\\*
& Run total  & 4.204 & 1.876 & 42.557 & 5,102.941 & ***\\*
\midrule
\textsc{ca-AstroPh} & LB & 87,632 & 87,632 & 87,632 & 91,260.0 & 91,260.0\\*
& UB & 176,099 & 162,552 & 162,550 & 175,540.0 & 175,078.0\\*
$n =$ 18,772 & Ratio  & 2.01 & 1.855 & 1.855 & 1.924 & 1.918\\* [2mm]
$m =$ 198,050 & Run LB  & 0.073 & 0.074 & 0.075 & 41.889 & 132.727\\*
& Run round  & 0.155 & 0.159 & 0.570 & 0.043 & 0.458\\*
& Run total  & 0.228 & 0.233 & 0.645 & 41.933 & 133.185\\*
\midrule
\textsc{ca-CondMat} & LB & 39,764 & 39,764 & 39,764 & 42,019.0 & 42,019.0\\*
& UB & 80,464 & 75,133 & 75,127 & 81,989.0 & 80,860.0\\*
$n =$ 23,133 & Ratio  & 2.024 & 1.889 & 1.889 & 1.951 & 1.924\\* [2mm]
$m =$ 93,439 & Run LB  & 0.023 & 0.022 & 0.021 & 7.961 & 18.080\\*
& Run round  & 0.158 & 0.132 & 0.532 & 0.067 & 0.161\\*
& Run total  & 0.181 & 0.154 & 0.553 & 8.028 & 18.241\\*
\midrule
\textsc{ca-GrQc} & LB & 4,789 & 4,789 & 4,789 & 5,196.0 & 5,196.0\\*
& UB & 9,336 & 8,424 & 8,424 & 9,046.0 & 8,528.0\\*
$n =$ 5,242 & Ratio  & 1.949 & 1.759 & 1.759 & 1.741 & 1.641\\* [2mm]
$m =$ 14,484 & Run LB  & 0.005 & 0.004 & 0.005 & 0.318 & 1.038\\*
& Run round  & 0.037 & 0.018 & 0.119 & 0.018 & 0.054\\*
& Run total  & 0.042 & 0.022 & 0.123 & 0.336 & 1.092\\*
\midrule
\textsc{ca-HepPh} & LB & 37,602 & 37,602 & 37,602 & 41,147.5 & 41,147.5\\*
& UB & 64,101 & 55,102 & 55,141 & 61,117.0 & 60,556.0\\*
$n =$ 12,008 & Ratio  & 1.705 & 1.465 & 1.466 & 1.485 & 1.472\\* [2mm]
$m =$ 118,489 & Run LB  & 0.307 & 0.322 & 0.314 & 32.220 & 110.161\\*
& Run round  & 0.111 & 0.040 & 2.110 & 0.031 & 0.132\\*
& Run total  & 0.418 & 0.362 & 2.424 & 32.251 & 110.293\\*
\midrule
\textsc{ca-HepTh} & LB & 10,855 & 10,855 & 10,855 & 11,516.5 & 11,516.5\\*
& UB & 21,806 & 21,006 & 21,006 & 22,538.0 & 21,870.0\\*
$n =$ 9,877 & Ratio  & 2.009 & 1.935 & 1.935 & 1.957 & 1.899\\* [2mm]
$m =$ 25,973 & Run LB  & 0.005 & 0.005 & 0.005 & 0.815 & 2.033\\*
& Run round  & 0.112 & 0.056 & 0.199 & 0.026 & 0.069\\*
& Run total  & 0.117 & 0.061 & 0.204 & 0.841 & 2.101\\*
\midrule
\textsc{cit-HepPh} & LB & 208,953 & 208,953 & 208,953 & 210,366.0 & --\\*
& UB & 418,074 & 417,075 & 417,069 & 420,702.0 & --\\*
$n =$ 34,546 & Ratio  & 2.001 & 1.996 & 1.996 & 2.0 & --\\* [2mm]
$m =$ 420,877 & Run LB  & 0.097 & 0.095 & 0.100 & 123.203 & --\\*
& Run round  & 0.247 & 0.049 & 0.908 & 0.008 & --\\*
& Run total  & 0.345 & 0.145 & 1.008 & 123.210 & --\\*
\midrule
\textsc{cit-HepTh} & LB & 174,612 & 174,612 & 174,612 & 175,975.5 & --\\*
& UB & 349,413 & 348,385 & 348,383 & 351,817.0 & --\\*
$n =$ 27,770 & Ratio  & 2.001 & 1.995 & 1.995 & 1.999 & --\\* [2mm]
$m =$ 352,285 & Run LB  & 0.089 & 0.143 & 0.100 & 141.213 & --\\*
& Run round  & 0.187 & 0.033 & 0.946 & 0.013 & --\\*
& Run total  & 0.276 & 0.176 & 1.046 & 141.226 & --\\*
\midrule
\textsc{cit-Patents} & LB & 8,141,691 & 8,141,691 & 8,141,691 & *** & ***\\*
& UB & 16,288,581 & 16,267,462 & 16,267,460 & *** & ***\\*
$n =$ 3,774,768 & Ratio  & 2.001 & 1.998 & 1.998 & *** & ***\\* [2mm]
$m =$ 16,518,947 & Run LB  & 6.331 & 6.125 & 6.969 & *** & ***\\*
& Run round  & 44.362 & 8.704 & 410.382 & *** & ***\\*
& Run total  & 50.693 & 14.830 & 417.351 & *** & ***\\*
\midrule
\textsc{com-Amazon} & LB & 426,524 & 426,524 & 426,524 & 455,045.5 & ***\\*
& UB & 859,596 & 834,788 & 834,693 & 897,651.0 & ***\\*
$n =$ 334,863 & Ratio  & 2.015 & 1.957 & 1.957 & 1.973 & ***\\* [2mm]
$m =$ 925,872 & Run LB  & 0.243 & 0.234 & 0.282 & 472.103 & ***\\*
& Run round  & 2.812 & 0.871 & 26.045 & 0.213 & ***\\*
& Run total  & 3.056 & 1.104 & 26.327 & 472.316 & ***\\*
\midrule
\textsc{com-DBLP} & LB & 404,019 & 404,019 & 404,019 & 449,921.5 & ***\\*
& UB & 825,495 & 737,218 & 737,172 & 853,743.0 & ***\\*
$n =$ 317,080 & Ratio  & 2.043 & 1.825 & 1.825 & 1.898 & ***\\* [2mm]
$m =$ 1,049,866 & Run LB  & 0.344 & 0.374 & 0.396 & 631.427 & ***\\*
& Run round  & 2.811 & 1.937 & 41.965 & 0.985 & ***\\*
& Run total  & 3.155 & 2.311 & 42.362 & 632.412 & ***\\*
\midrule
\textsc{com-LiveJournal} & LB & 16,772,636 & 16,772,636 & 16,772,636 & *** & ***\\*
& UB & 33,608,381 & 33,275,587 & 33,277,914 & *** & ***\\*
$n =$ 3,997,962 & Ratio  & 2.004 & 1.984 & 1.984 & *** & ***\\* [2mm]
$m =$ 34,681,189 & Run LB  & 17.002 & 17.034 & 25.346 & *** & ***\\*
& Run round  & 48.159 & 23.847 & 1,742.355 & *** & ***\\*
& Run total  & 65.161 & 40.880 & 1,767.701 & *** & ***\\*
\midrule
\textsc{com-Orkut} & LB & 58,505,482 & 58,505,482 & 58,505,482 & *** & ***\\*
& UB & 117,026,574 & 116,953,503 & 116,953,361 & *** & ***\\*
$n =$ 3,072,441 & Ratio  & 2.0 & 1.999 & 1.999 & *** & ***\\* [2mm]
$m =$ 117,185,083 & Run LB  & 122.594 & 121.138 & 135.962 & *** & ***\\*
& Run round  & 215.619 & 17.428 & 17,305.204 & *** & ***\\*
& Run total  & 338.213 & 138.566 & 17,441.166 & *** & ***\\*
\midrule
\textsc{com-Youtube} & LB & 1,429,785 & 1,429,785 & 1,429,785 & *** & ***\\*
& UB & 2,860,915 & 2,855,070 & 2,855,066 & *** & ***\\*
$n =$ 1,134,890 & Ratio  & 2.001 & 1.997 & 1.997 & *** & ***\\* [2mm]
$m =$ 2,987,624 & Run LB  & 0.775 & 0.788 & 1.024 & *** & ***\\*
& Run round  & 9.276 & 3.136 & 63.352 & *** & ***\\*
& Run total  & 10.051 & 3.924 & 64.376 & *** & ***\\*
\midrule
\textsc{email-Enron} & LB & 84,385 & 84,385 & 84,385 & 87,861.0 & 87,861.0\\*
& UB & 169,567 & 165,774 & 165,765 & 174,693.0 & 172,762.0\\*
$n =$ 36,692 & Ratio  & 2.009 & 1.964 & 1.964 & 1.988 & 1.966\\* [2mm]
$m =$ 183,831 & Run LB  & 0.041 & 0.042 & 0.042 & 113.530 & 339.700\\*
& Run round  & 0.257 & 0.266 & 1.067 & 0.069 & 0.238\\*
& Run total  & 0.298 & 0.307 & 1.109 & 113.599 & 339.938\\*
\midrule
\textsc{email-EuAll} & LB & 174,651 & 174,651 & 174,651 & *** & ***\\*
& UB & 349,298 & 349,296 & 349,296 & *** & ***\\*
$n =$ 265,214 & Ratio  & 2.0 & 2.0 & 2.0 & *** & ***\\* [2mm]
$m =$ 364,481 & Run LB  & 0.088 & 0.087 & 0.098 & *** & ***\\*
& Run round  & 1.822 & 0.348 & 5.963 & *** & ***\\*
& Run total  & 1.911 & 0.435 & 6.061 & *** & ***\\*
\midrule
\textsc{loc-Brightkite} & LB & 101,924 & 101,924 & 101,924 & 106,429.0 & ***\\*
& UB & 204,011 & 203,016 & 203,016 & 212,545.0 & ***\\*
$n =$ 58,228 & Ratio  & 2.002 & 1.992 & 1.992 & 1.997 & ***\\* [2mm]
$m =$ 214,078 & Run LB  & 0.042 & 0.041 & 0.043 & 76.008 & ***\\*
& Run round  & 0.387 & 0.197 & 1.306 & 0.018 & ***\\*
& Run total  & 0.428 & 0.238 & 1.349 & 76.026 & ***\\*
\midrule
\textsc{loc-Gowalla} & LB & 456,499 & 456,499 & 456,499 & *** & ***\\*
& UB & 914,484 & 907,916 & 907,897 & *** & ***\\*
$n =$ 196,591 & Ratio  & 2.003 & 1.989 & 1.989 & *** & ***\\* [2mm]
$m =$ 950,327 & Run LB  & 0.203 & 0.245 & 0.213 & *** & ***\\*
& Run round  & 1.539 & 0.701 & 7.949 & *** & ***\\*
& Run total  & 1.742 & 0.946 & 8.162 & *** & ***\\*
\midrule
\textsc{roadNet-CA} & LB & 1,275,870 & 1,275,870 & 1,275,870 & 1.379125e6 & ***\\*
& UB & 2,556,485 & 2,536,702 & 2,536,702 & 2.745368e6 & ***\\*
$n =$ 1,971,281 & Ratio  & 2.004 & 1.988 & 1.988 & 1.991 & ***\\* [2mm]
$m =$ 2,766,607 & Run LB  & 0.547 & 0.985 & 0.406 & 6,263.413 & ***\\*
& Run round  & 16.925 & 7.356 & 188.404 & 1.377 & ***\\*
& Run total  & 17.472 & 8.341 & 188.810 & 6,264.791 & ***\\*
\midrule
\textsc{roadNet-PA} & LB & 711,585 & 711,585 & 711,585 & 769,226.0 & ***\\*
& UB & 1,425,699 & 1,414,908 & 1,414,908 & 1.531777e6 & ***\\*
$n =$ 1,090,920 & Ratio  & 2.004 & 1.988 & 1.988 & 1.991 & ***\\* [2mm]
$m =$ 1,541,898 & Run LB  & 0.193 & 0.195 & 0.189 & 1,367.728 & ***\\*
& Run round  & 9.130 & 2.911 & 70.004 & 0.570 & ***\\*
& Run total  & 9.323 & 3.106 & 70.193 & 1,368.297 & ***\\*
\midrule
\textsc{roadNet-TX} & LB & 883,250 & 883,250 & 883,250 & 958,145.5 & ***\\*
& UB & 1,769,957 & 1,755,407 & 1,755,407 & 1.9056e6 & ***\\*
$n =$ 1,393,383 & Ratio  & 2.004 & 1.987 & 1.987 & 1.989 & ***\\* [2mm]
$m =$ 1,921,660 & Run LB  & 0.245 & 0.245 & 0.239 & 3,041.782 & ***\\*
& Run round  & 11.510 & 4.519 & 106.177 & 0.749 & ***\\*
& Run total  & 11.755 & 4.764 & 106.416 & 3,042.530 & ***\\*
\midrule
\textsc{soc-Epinions1} & LB & 197,337 & 197,337 & 197,337 & 202,521.0 & ***\\*
& UB & 394,809 & 394,032 & 394,031 & 404,825.0 & ***\\*
$n =$ 75,888 & Ratio  & 2.001 & 1.997 & 1.997 & 1.999 & ***\\* [2mm]
$m =$ 405,740 & Run LB  & 0.091 & 0.091 & 0.103 & 808.772 & ***\\*
& Run round  & 0.516 & 0.413 & 1.706 & 0.035 & ***\\*
& Run total  & 0.607 & 0.504 & 1.809 & 808.807 & ***\\*
\midrule
\textsc{soc-LiveJournal1} & LB & 20,681,921 & 20,681,921 & 20,681,921 & *** & ***\\*
& UB & 41,452,776 & 40,988,289 & 40,999,181 & *** & ***\\*
$n =$ 4,847,571 & Ratio  & 2.004 & 1.982 & 1.982 & *** & ***\\* [2mm]
$m =$ 42,851,237 & Run LB  & 41.378 & 38.295 & 21.741 & *** & ***\\*
& Run round  & 97.697 & 47.469 & 3,212.986 & *** & ***\\*
& Run total  & 139.075 & 85.764 & 3,234.728 & *** & ***\\*
\midrule
\textsc{soc-Slashdot0811} & LB & 229,500 & 229,500 & 229,500 & 234,429.5 & ***\\*
& UB & 459,048 & 458,653 & 458,653 & 468,752.0 & ***\\*
$n =$ 77,360 & Ratio  & 2.0 & 1.998 & 1.998 & 2.0 & ***\\* [2mm]
$m =$ 469,180 & Run LB  & 0.101 & 0.101 & 0.104 & 488.743 & ***\\*
& Run round  & 0.698 & 0.186 & 1.594 & 0.024 & ***\\*
& Run total  & 0.800 & 0.287 & 1.697 & 488.767 & ***\\*
\midrule
\textsc{soc-Slashdot0902} & LB & 247,059 & 247,059 & 247,059 & 251,944.0 & ***\\*
& UB & 494,192 & 493,777 & 493,778 & 503,815.0 & ***\\*
$n =$ 82,168 & Ratio  & 2.0 & 1.999 & 1.999 & 2.0 & ***\\* [2mm]
$m =$ 504,230 & Run LB  & 0.103 & 0.103 & 0.112 & 604.756 & ***\\*
& Run round  & 0.560 & 0.190 & 1.844 & 0.030 & ***\\*
& Run total  & 0.663 & 0.293 & 1.956 & 604.786 & ***\\*
\midrule
\textsc{web-BerkStan} & LB & 3,101,047 & 3,101,047 & 3,101,047 & *** & ***\\*
& UB & 6,217,043 & 6,148,769 & 6,148,776 & *** & ***\\*
$n =$ 685,230 & Ratio  & 2.005 & 1.983 & 1.983 & *** & ***\\* [2mm]
$m =$ 6,649,470 & Run LB  & 3.144 & 3.179 & 2.982 & *** & ***\\*
& Run round  & 6.065 & 1.077 & 75.884 & *** & ***\\*
& Run total  & 9.209 & 4.256 & 78.866 & *** & ***\\*
\midrule
\textsc{web-Google} & LB & 1,889,936 & 1,889,936 & 1,889,936 & *** & ***\\*
& UB & 3,820,726 & 3,652,072 & 3,652,510 & *** & ***\\*
$n =$ 916,428 & Ratio  & 2.022 & 1.932 & 1.933 & *** & ***\\* [2mm]
$m =$ 4,322,051 & Run LB  & 2.109 & 2.534 & 2.151 & *** & ***\\*
& Run round  & 9.049 & 1.851 & 125.882 & *** & ***\\*
& Run total  & 11.158 & 4.385 & 128.034 & *** & ***\\*
\midrule
\textsc{web-NotreDame} & LB & 414,253 & 414,253 & 414,253 & *** & ***\\*
& UB & 829,768 & 820,767 & 820,771 & *** & ***\\*
$n =$ 325,729 & Ratio  & 2.003 & 1.981 & 1.981 & *** & ***\\* [2mm]
$m =$ 1,090,108 & Run LB  & 0.479 & 0.472 & 0.642 & *** & ***\\*
& Run round  & 2.574 & 0.593 & 21.478 & *** & ***\\*
& Run total  & 3.053 & 1.064 & 22.120 & *** & ***\\*
\midrule
\textsc{web-Stanford} & LB & 948,859 & 948,859 & 948,859 & *** & ***\\*
& UB & 1,903,048 & 1,877,270 & 1,877,228 & *** & ***\\*
$n =$ 281,903 & Ratio  & 2.006 & 1.978 & 1.978 & *** & ***\\* [2mm]
$m =$ 1,992,636 & Run LB  & 0.961 & 1.035 & 0.911 & *** & ***\\*
& Run round  & 2.476 & 0.587 & 13.140 & *** & ***\\*
& Run total  & 3.437 & 1.621 & 14.052 & *** & ***\\*
\midrule
\textsc{wiki-Talk} & LB & 2,313,080 & 2,313,080 & 2,313,080 & *** & ***\\*
& UB & 4,626,181 & 4,626,014 & 4,626,014 & *** & ***\\*
$n =$ 2,394,385 & Ratio  & 2.0 & 2.0 & 2.0 & *** & ***\\* [2mm]
$m =$ 4,659,565 & Run LB  & 1.546 & 1.600 & 1.557 & *** & ***\\*
& Run round  & 22.358 & 1.233 & 103.372 & *** & ***\\*
& Run total  & 23.904 & 2.834 & 104.929 & *** & ***\\*
\midrule
\textsc{wiki-topcats} & LB & 12,689,197 & 12,689,197 & 12,689,197 & *** & ***\\*
& UB & 25,380,869 & 25,369,215 & 25,369,203 & *** & ***\\*
$n =$ 1,791,489 & Ratio  & 2.0 & 1.999 & 1.999 & *** & ***\\* [2mm]
$m =$ 25,444,207 & Run LB  & 8.091 & 8.290 & 8.534 & *** & ***\\*
& Run round  & 19.867 & 1.737 & 104.694 & *** & ***\\*
& Run total  & 27.958 & 10.026 & 113.228 & *** & ***\\*
\bottomrule
\end{longtable}